\documentclass[review,onefignum,onetabnum]{siamart171218}
\usepackage{geometry}
\usepackage{tikz} 
\usepackage{tikz-cd}
\usetikzlibrary{decorations.shapes}
\usetikzlibrary{shadows, positioning, calc}
\tikzset{decorate sep/.style 2 args={decorate,decoration={shape backgrounds,shape=circle,shape size=#1,shape sep=#2}}}

\usepackage{lipsum}
\usepackage{amsfonts}
\usepackage{graphicx}
\usepackage{epstopdf}
\usepackage{algorithmic}
\usepackage{amsmath}
\usepackage{mathrsfs}
\usepackage{caption}
\usepackage{subcaption}
\ifpdf
  \DeclareGraphicsExtensions{.eps,.pdf,.png,.jpg}
\else
  \DeclareGraphicsExtensions{.eps}
\fi

\newsiamremark{remark}{Remark}
\newsiamremark{hypothesis}{Hypothesis}
\crefname{hypothesis}{Hypothesis}{Hypotheses}
\newsiamthm{claim}{Claim}

\headers{Numerical analysis of Exchange Option}{Kevin S. Zhang, Traian A. Pirvu}

\title{Numerical Simulation of Exchange Option with Finite Liquidity: Controlled Variate Model\thanks{Preprint.
\funding{This work was funded by NSERC grant 5-36700.}}}

\author{
Kevin S. Zhang\thanks{Department of Mathematics and Statistics, McMaster University, Hamilton, ON, Canada.}
\and 
Traian A. Pirvu\footnotemark[3]\thanks{Department of Mathematics and Statistics, McMaster University, Hamilton, ON, Canada.}
}

\usepackage{amsopn}
\usepackage{bbm}

\makeatletter
\newcommand*{\addFileDependency}[1]{% argument=file name and extension
  \typeout{(#1)}% latexmk will find this if $recorder=0 (however, in that case, it will ignore #1 if it is a .aux or .pdf file etc and it exists! if it doesn't exist, it will appear in the list of dependents regardless)
  \@addtofilelist{#1}% if you want it to appear in \listfiles, not really necessary and latexmk doesn't use this
  \IfFileExists{#1}{}{\typeout{No file #1.}}% latexmk will find this message if #1 doesn't exist (yet)
}
\makeatother

\ifpdf
\hypersetup{
  pdftitle={Numerical Simulation of Exchange Option with finite liquidity: control variate model},
  pdfauthor={Kevin Zhang, Traian A. Pirvu}
}
\fi

\begin{document}

\maketitle

% REQUIRED
\begin{abstract}
  In this paper we develop numerical pricing methodologies for European style Exchange Options written on a pair of correlated assets, in a market with finite liquidity. In contrast to the standard multi-asset Black-Scholes framework, trading in our market model has a direct impact on the asset's price. The price impact is incorporated into the dynamics of the first asset through a specific trading strategy, as in large trader liquidity model. Two-dimensional Milstein scheme is implemented to simulate the pair of assets prices. The option value is numerically estimated by Monte Carlo with the Margrabe option as controlled variate. Time complexity of these numerical schemes are included. Finally, we provide a deep learning framework to implement this model effectively in a production environment.
\end{abstract}

% REQUIRED

% REQUIRED

\section{Introduction}
The Black-Scholes (BS) model was truly a breakthrough for pricing single asset options. It assumes participants operate in a perfectly liquid, friction-less and complete market. In practice, one or more of these assumptions are violated. When the liquidity restriction is relaxed, trading will impact the price of the underlying assets. Wilmott (2000) \cite{Wilmott} was one of the pioneers of these price impact models. He considered price impacts depending upon different trading strategies such as buy and hold, limit order and portfolio optimization. To account for price impact, Liu and Yong (2005) \cite{LiunYong} included an additional term in the asset price stochastic differential equation (SDE). This inclusion indirectly adds a valuation adjustment to the price of the option. Such an adjustment stems from a lack of liquidity, and may be classified as \textit{liquidity valuation adjustment} (LVA). Various non-linear BS-like partial differential equations (PDE), capturing the resulting price impact from trading have been studied \cite{Glover, Dyshaev, Ahmadian, Arenas}. All these models share the similarity of being single-asset LVA models.
\par
\textit{Exchange Options} provide the utility of exchanging one asset for another. Under the BS assumption for binary asset markets, Margrabe (1978) \cite{Margrabe} derived a closed form solution for the price of Exchange Options. The Exchange Option plays an essential role in currency markets. The \textit{Foreign Exchange} (FX) \textit{Option} is an Exchange Option where the assets are currencies. A common concern is raised when one considers the interaction between liquid and illiquid currencies. A trader might ask, ``How reliable is the price of a 3-month European style USD/UAH (Ukrainian Hryvnia, an infrequently traded currency) FX Option?". In this work, we are interested in these type of scenarios. Recent studies on Exchange Options, such as \cite{Alos1,Alos2,Siraj,Hainaut}, exhibit deviation from the assumptions of BS. The aforementioned studies predominately involve stochastic volatility models. Similar to Exchange Options, studies on Spread Option pricing have been conducted in the presence of full or partial price impact \cite{Pirvu, Pirvu2}.
\par
In this paper, we consider a binary-asset market with a single illiquid asset. Under this consideration, we construct a price impact model, called the \textit{finite liquidity market model} (FLMM). The model is a system of SDEs, one for each asset. The liquid asset is unchanged, the illiquid is modified to incorporate the resulting price impact from trading. Existence and uniqueness conditions on the SDES are established for the FLMM (see \autoref{appendix}). By replicating a portfolio, We derive the partial differential equation (PDE) characterization of option prices. Further, we consider a market consisting of market makers, who trade by Delta Hedging. %Next, we specify the price impact function and consider a scenario where market participants trades by Delta hedging, thought of as the strategy of large traders who impacts the market. 
We utilize the Milstein method and simulate the FLMM SDEs as inspired by \cite{Giles2,Desmond}. The Margrabe Exchange Option is used as the control variate for our Monte Carlo (MC) pricing of the option. Motivated by \cite{Ferguson,Culkin}, we apply deep feed-forward network to our MC pricing engine and achieves accurate high speed pricing.
\par
The remainder of the content written in this paper is organized in the following sections. Section \ref{sec:Model} discusses the model framework. In Section \ref{sec:Delta}, we analyze the price impact effect when majority of the market participants implement Delta Hedging. In Section \ref{sec:Num}, we apply Milstein's method to simulate the path-wise price and sensitivity. Subsequently, we deploy control variate MC for estimation. Section \ref{sec:deep} contains the methodology of \textit{Deeply Learning Derivative} for Exchange Option with price impact. In Section \ref{sec:conclusions}, we make some concluding statements for the readers. The last Section is an Appendix containing the proofs of our results.

\section{Model Framework}
\label{sec:Model}
In this section we describe the dynamics of FLMM. There is a filtered probability space $\big(\Omega,\mathbb{P},\mathscr{F}(t)\big)$ that satisfies the usual conditions. There are two risky assets whose prices are assumed to be a two-dimensional correlated I\^to process $\mathbf{S}(t)=\big(S_1(t),S_2(t)\big)$. There is also a risk-free asset $D(t).$ The uncertainty in this model is driven by a two-dimensional independent Brownian Motion $\mathbf{W}(t)=\big(W_1(t),W_2(t)\big)$. The system of SDEs which captures the asset price dynamics can be illustrated as follows:
\begin{align}
\label{finiteog}
&\frac{dS_1(t)}{S_1(t)}=\mu_1(t)dt+\sigma_{1}dW_1(t)+\lambda\big(t,S_1(t),S_2(t)\big)df\big(t,S_1(t),S_2(t)\big),\nonumber
\\
&\frac{dS_2(t)}{S_2(t)}=\mu_2(t)dt+\sigma_{2}\rho{}dW_1(t)+\sigma_{2}\sqrt{1-\rho^2}dW_2(t),
\\
&\frac{dD(t)}{D(t)}=-rdt,\nonumber
\end{align}
where $\mu_i(t)$, $\sigma_{i}$, $\rho$ are the drift process, volatility and correlation of each I\^to Process respectively. The novelty here is the term $\lambda\big(t,S_1(t),S_2(t)\big)df\big(t,S_1(t),S_2(t)\big),$ and it represents the price impact $\lambda(t,s_1,s_2)$ from a trading strategy $f(t,s_1,s_2).$ We will assumed the price impact is always non-negative, that is $\lambda(t,s_1,s_2)\geq{0}.$ Let us point out the two-dimensional market model used by Margrabe (1978) \cite{Margrabe} is a special case of this model when $\lambda(t, s_1,s_2)=0$. %Naturally, it makes sense to use the impact-less Exchange Option as the control variate of our MC %estimator.
\par
We plan to obtain a canonical SDE of Asset $1,$ and this will allow for a better understanding of the model's dynamics. In order to achieve this, we first apply It\^o's Theorem to compute the following differential $df\big(t,S_1(t),S_2(t)\big).$ Then, we isolate the $dS_1(t)$ terms, and compute the following quadratic/cross-variation terms: $dS_1(t) dS_1(t)$, $dS_1(t) dS_2(t)$ and $dS_2(t) dS_2(t)$. By doing so we arrive at:

%It actually only take some basic algebra and the use of It\^o's Theorem to arrange the SDE system \eqref{finiteog} to the following format:
    \begin{align}\label{FLMMSDE}
    dS_1(t)&=\bar{\mu}_1\big(\mathbf{S}(t)\big)dt+\bar{\sigma}_{11}\big(\mathbf{S}(t)\big)dW_1(t)+\bar{\sigma}_{12}\big(\mathbf{S}(t)\big)dW_2(t),
    \end{align}
    where the drift and diffusion functions are:
    \begin{align*}
    &\bar{\mu}_1(t,s_1,s_2)=\frac{1}{1-\lambda{}f_{s_1}}\Big(\mu_1s_1+\lambda{}f_{t}+s_2\mu_2\lambda{}f_{s_2}+\frac{f_{s_1s_2}(\rho\sigma_1\sigma_2s_1s_2+\sigma_2^2s_2^2\lambda{}f_{s_2})}{1-\lambda{}f_{s_1}}
    \\
    &\qquad\qquad\quad+\frac{f_{s_1s_1}(\sigma_1^2s_1^2+\sigma^2_2s_2^2\lambda^2f_{s_2}^2+2\rho\sigma_1\sigma_2s_1s_2\lambda{}f_{s_2})}{2\big(1-\lambda{}f_{s_1}\big)^2}+\frac{\sigma_2^2s_2^2f_{s_2s_2}}{2}\Big),
    \\
    &\bar{\sigma}_{11}(t,s_1,s_2)=\frac{\sigma_1s_1}{1-\lambda{}f_{s_1}},\qquad\bar{\sigma}_{12}(t,s_1,s_2)=\frac{\sigma_2s_2\lambda{}f_{s_2}}{1-\lambda{}f_{s_1}}.
    \end{align*}
\par
With the model dynamics in hand, we can determine the requirements for the SDE driving $S_1$ to have a unique solution. In classical literature on SDE such as Oksendal (1992) \cite{Oksendal}, there are classical theorems for the existence and uniqueness of different kinds (strong, weak) solutions. The following theorem provides sufficient conditions for the existence and uniqueness of FLMM SDEs:

\begin{theorem}[\textbf{Finite Liquidity Existence and Uniqueness Theorem I}]
\label{theorem:FLMMexist}
Under the assumptions $(1)$ to $(6)$ of (\ref{FLEUT1}), the SDE of
$S_1$ in \eqref{FLMMSDE} has a unique strong solution.
\end{theorem}
\begin{proof}
Please refer to the Appendix Section \ref{FLEUT1}.
\end{proof}

%\par
% we are interested in pricing and hedging under the FLMM framework, thus it is only natural to discover a risk-neutral measure. It turns out the risk-neutral measure is unique, thus FLMM is complete. The following theorem characterizes the risk-neutral measure:

%\begin{theorem}[\textbf{Finite Liquidity Risk-Neutral Measure}]
%\label{thm:RNM}
%There exists a unique risk-neutral measure for the finite liquidity model given by the standard Girsanov machinery with the vector-valued market price of risk generator process:
%    \begin{gather*}
%    \mathbf{\Theta}\big(t,S_1(t),S_2(t)\big)=\frac{1}{\sigma_2(\sqrt{1-\rho^2}\bar{\sigma}_{11}-\rho\bar{\sigma}_{12})}
%    \begin{bmatrix} 
%    \sqrt{1-\rho^2}\sigma_2&-\bar{\sigma}_{12}\\
%    -\rho\sigma_2&\bar{\sigma}_{11}
%    \end{bmatrix}
%    \begin{bmatrix} 
%    \bar{\mu}_1-r\\
%    \mu_2-r
%    \end{bmatrix}.
%    \end{gather*}
%    \end{theorem}
%    \begin{proof}
%Please see \ref{FLRNM} in the Appendix Section.
%\end{proof}

\par
The replicating portfolio argument is fundamental to the derivations of BS equation. The replication argument in Chapter 4.5 of Shreve (2004) \cite{ShreveII} can be modified to replicate the option within FLMM framework. The portfolio used for replication will have two assets and one cash account. The resulting equation will be a linear BS-like PDE of the parabolic family:
    \[\left\{
      \begin{aligned}
      \label{FLMMPDE}
        rV&=V_t+rs_1V_{s_1}+rs_2V_{s_2}+\frac{V_{s_1s_2}}{1-\lambda{}f_{s_1}}\big(\rho\sigma_1\sigma_2s_1s_2+\lambda{}f_{s_2}\sigma_2^2s_2^2\big)   \cr
        &+\frac{V_{s_1s_1}}{2(1-\lambda{}f_{s_1})^2}\big(\sigma_1^2s_1^2+\lambda^2f_{s_2}^2\sigma^2_2s_2^2+2\lambda{}f_{s_2}\rho\sigma_1\sigma_2s_1s_2\big)+\frac{1}{2}V_{s_2s_2}\sigma_2^2s_2^2,   \cr
        V(T,s_1,s_2)&=h(s_1,s_2),\quad\text{with $0<s_1,s_2<\infty$, $0\leq{t}\leq{T}$},
      \end{aligned}\right.
    \]
    where $h(s_1,s_2)$ is a general payoff function. Existence results in Chapter 4 of Friedman (1975) \cite{Friedman} yield a unique classical solution for this PDE, granted $1-\lambda{}f_{s_1}$ satisfies condition $(3)$ of Theorem \ref{theorem:FLMMexist}.
\par
Feynman-Kac formula allows that the solutions for this PDE to be represented as a conditional expectations. As a by product of Feynman-Kac, we will discover an induced risk-neutral measure $\widetilde{\mathbbm{P}}$. Under this measure, we have the pricing formula:
\begin{gather}\label{FLRNPF}
V\big(t,s_1,s_2\big)=\widetilde{\mathbb{E}}^{t,s_1,s_2}[e^{-r(T-t)}V\big(T,S_1(T),S_2(T)\big)].
\end{gather}

\section{Analysis of Replication of Exchange Option by Delta Hedging as Price Impact}
\label{sec:Delta}
In this section, we show that FLMM has a unique strong solution for a specific choice of price impact $\lambda\big(t,S_1(t)\big)df\big(t,S_1(t),S_2(t)\big)$. There have been numerous studies in the past focused on price impacts from trading. For example, Liu and Yong (2005) \cite{LiunYong} studied a price impact model for single asset options. Pirvu et al. (2014) \cite{Pirvu} also studied a price impact model for spread option. In this paper, we adopt the following price impact function:
    \begin{align}\label{impact}
    \bar{\lambda}\big(t,s_1\big)=
    \begin{cases}
    \epsilon\big(1-e^{-\beta(T-t)^{\frac{3}{2}}}\big)\quad&\text{if}\quad\underline{S_1}\leq{}s_1\leq\overline{S}_1,
    \\
    0\quad&\text{otherwise},
    \end{cases}
    \end{align}
where $\underline{S_1}$ and $\overline{S}_1$ represents a trading floor and cap of the asset respectively. This cause the trading price impact to be truncated within the floor and cap. As for the other parameters, $\epsilon$ is the price impact per share, and $\beta$ is a decaying constant. 
%For the rest of this chapter, $\ushort{S_1}$ and $\overline{S}_1$ will be set to $60$\% and $140$\% of the current fair market value of the asset. $\epsilon$ will be set to $0.04$, and $\beta$ will be set to $100$. 
It is important to emphasize that $\bar{\lambda}(t,s_1)$ will be employed for numerical approximation. The theoretical $\lambda(t,s_1)$ should be a function with bounded derivative, that is obtained through standard mollifying $\bar{\lambda}(t,s_1)$.
\par
Delta hedging is a strategy adopted by many big financial institutions to reduce their option portfolio's exposure against movements in the underlying assets. In this paper, we assume majority of the market participants implement Delta hedging with the Delta of the impact-less Exchange Option. Therefore, we choose the trading strategy function to be $\Delta_1(t)$ of Margrabe's option, that is $f\big(t,s_1,s_2\big)=\Delta_1(t)$. The closed form expression for $\Delta_{1}$ can be found in the Appendix Section (\ref{sec:greek}).

As a result, the drift and diffusion functions in \eqref{FLMMSDE} have the following dynamics:
    \begin{align}\label{FLMMSDEexchange}
    &\widetilde{\mu}_1(t,s_1,s_2)=\frac{1}{1-\lambda\Gamma_{11}}\Big(\mu_1s_1+\lambda{}Chm_1+\mu_2s_2\lambda\Gamma_{12}
    +\frac{Spd_{112}(\rho\sigma_1\sigma_2s_1s_2+\sigma_2^2s_2^2\lambda\Gamma_{12})}{1-\lambda\Gamma_{11}}\nonumber
    \\
    &\qquad\qquad\quad+\frac{Spd_{111}(\sigma_1^2s_1^2+\sigma^2_2s_2^2\lambda^2\Gamma_{12}^2+2\rho\sigma_1\sigma_2s_1s_2\lambda\Gamma_{12})}{2(1-\lambda\Gamma_{11})^2}+\frac{\sigma_2^2s_2^2Spd_{122}}{2}\Big),\nonumber
    \\
    &\widetilde{\sigma}_{11}(t,s_1,s_2)=\frac{\sigma_1s_1}{1-\lambda\Gamma_{11}},\qquad\widetilde{\sigma}_{12}(t,s_1,s_2)=\frac{\sigma_2s_2\lambda\Gamma_{12}}{1-\lambda\Gamma_{11}}.\nonumber
    \end{align}
Here $Chm$, $\Gamma$ and $Spd$ are higher order Greeks of Magrabe's option derived from Margrabe's formula. All the Greek formulas are given in the Appendix section \ref{sec:greek}.

\begin{theorem}[\textbf{Existence and Uniqueness of Finite Liquidity Market \\Model SDE II}]
\label{theorem:ExistFLMM2}
\qquad\qquad\qquad\qquad\qquad
\par
The SDE of $S_1$ with drift and diffusion function of (\ref{sec:greek}) has a unique strong solution.
\end{theorem}
\begin{proof}
Please refer to Appendix \ref{FLEUT1} for the proof.
\end{proof}

\section{Numerical Simulations}
\label{sec:Num}
\par
In this section, our first objective is to simulate the FLMM assets by applying the Milstein Algorithm. Once we have the asset processes, we can use the results in our control variate MC estimator to price the Exchange Option with price impact. As a naming convention for our analysis, we refer to the number of points $M$ used to generate the stochastic assets as ``path dimension". The amount of asset paths $N$ used in the MC estimator will be referred to as ``space dimension".

\subsection{Milstein Scheme for Asset Price}
\par
Compared with the more well known Euler-Maruyama, Milstein is a second-order pathwise method for approximating SDE solutions. It was created by Mil’shtein G. N. (1975) \cite{Milstein}, this method retains the second order terms from I\^to Taylor expansion. For a 2-dimensional SDE system satisfied by the process $\mathbf{X}(t)=\big(X_1(t),X_2(t)\big),$ a second-order approximation of the solution is:
    \begin{align}
    \label{itoexpansion2}
    X_1(t)&\approx{}X_1(t_0)+\int_{t_0}^t\mu_1\big(\mathbf{X}(u)\big)du+\int_{t_0}^t\sigma_{11}\big(\mathbf{X}(u)\big)dW_1(u)+\int_{t_0}^t\sigma_{12}\big(\mathbf{X}(u)\big)dW_2(u)\nonumber
    \\
    &+\frac{1}{2}\sum_{j,k,l=1}^2\frac{\partial\sigma_{1j}}{\partial{}x_l}\sigma_{lk}\big(\mathbf{X}(t_0)\big)\big(\Delta{W}_j(t)\Delta{W}_k(t)+\rho_{jk}(t-t_0)-\mathcal{A}_{jk}(t_0,t)\big),\nonumber
    \\
    X_2(t)&\approx{}X_2(t_0)+\int_{t_0}^t\mu_2\big(\mathbf{X}(u)\big)du+\int_{t_0}^t\sigma_{21}\big(\mathbf{X}(u)\big)dW_1(u)+\int_{t_0}^t\sigma_{22}\big(\mathbf{X}(u)\big)dW_2(u)\nonumber
    \\
    &+\frac{1}{2}\sum_{j,k,l=1}^2\frac{\partial\sigma_{2j}}{\partial{}x_l}\sigma_{lk}\big(\mathbf{X}(t_0)\big)\big(\Delta{W}_j(t)\Delta{W}_k(t)+\rho_{jk}(t-t_0)-\mathcal{A}_{jk}(t_0,t)\big),\nonumber
    \end{align}
    According to Giles (2018) \cite{Giles2}, the term $A_{ij}(t_0,t)$ is the L\'evy Area between two the two driving Brownian motions. It's behavior is captured by following stochastic integral:
    \begin{align}\label{Larea}
    \mathcal{A}_{ij}(t_0,t)=\int^{t}_{t_0}\big(\Delta{W}_i(u)dW_j(u)-\Delta{W}_j(u)dW_i(u)\big).
    \end{align}
\par
Since we are only interested in pricing and hedging, it is advantageous to work under the risk-neutral measure. FLMM in \eqref{FLMMSDE} with the updated drift and diffusion functions of \eqref{FLMMSDEexchange} has the following dynamics under $\widetilde{\mathbbm{P}}$:
    \begin{align}\label{rnpde}
    dS_1(t)&=rS_1(t)dt+\widetilde{\sigma}_{11}\big(\mathbf{S}(t)\big)d\widetilde{W}_1(t)+\widetilde{\sigma}_{12}\big(\mathbf{S}(t)\big)d\widetilde{W}_2(t),\nonumber
    \\
    dS_2(t)&=rS_2(t)dt+\widetilde{\sigma}_{21}(t)d\widetilde{W}_1(t)+\widetilde{\sigma}_{22}(t)d\widetilde{W}_2(t),
    \\
    \frac{dD(t)}{D(t)}&=-rdt,\nonumber
    \end{align}
for simplicity, we set:
    \begin{align*}
    &\widetilde{\sigma}_{21}(t)=\sigma_2s_2\rho{},\qquad\widetilde{\sigma}_{22}(t)=\sigma_2s_2\sqrt{1-\rho^2}.
    \end{align*}
\par
The Milstein approximation for \eqref{rnpde} can be set up by following these procedures:
    \begin{enumerate}
    \label{milstein}
    \item Partition $[t,T]$ into $M$ equivalent intervals of length $\Delta{t}=\frac{T-t}{M}$.
    \item Set the initial values as $S_1(0)=s_1$ and $S_2(0)=s_2$.
    \item Sample $\{\Delta{}W_1(j),\Delta{}W_2(j)\}_{j=1,2,...M}$, where each \\$\{\Delta{}W_1(j),\Delta{}W_2(j)\}\sim\mathcal{N}_2(\mathbf{0},\Delta{t}I_2)$.
    \item Generate L\'evy Areas $\mathcal{A}_{ij}(0,\Delta{t})$.
    \item Recursively define:
    \begin{align}
    \label{milsteinscheme}
    S_1(m+1)&=S_1(m)+rS_1(m)\Delta{t}+\sum_{i=1}^2\widetilde{\sigma}_{1i}\big(\mathbf{S}(m)\big)\Delta{}W_i(m+1)+\frac{1}{2}\sum_{i,j,k=1}^2\frac{\partial{\widetilde{\sigma}}_{1i}}{\partial{s}_k}\nonumber
    \\
    &\times\widetilde{\sigma}_{kj}\big(\mathbf{S}(m)\big)\big(\Delta{}W_i(m+1)\Delta{}W_j(m+1)-\mathbbm{1}_{(i=j)}\Delta{t}-\mathcal{A}_{ij}\big),\nonumber
    \\
    S_2(m+1)&=S_2(m)+rS_2(m)\Delta{t}+\sum_{i=1}^2\widetilde{\sigma}_{2i}\big(\mathbf{S}(m)\big)\Delta{}W_i(m+1)+\frac{1}{2}\sum_{i,j,k=1}^2\frac{\partial{\widetilde{\sigma}}_{2i}}{\partial{s}_k}\nonumber
    \\
    &\times\widetilde{\sigma}_{kj}\big(\mathbf{S}(m)\big)
    \big(\Delta{}W_i(m+1)\Delta{}W_j(m+1)-\mathbbm{1}_{(i=j)}\Delta{t}-\mathcal{A}_{ij}\big).\nonumber
    \end{align} 
    \end{enumerate}
\par
There are many techniques to approximate the L\'evy Area, one of the simplest is to generate the stochastic integral \eqref{Larea} piece by piece. In this paper, we adopted an algorithm which closely resembles the method found in Scheicher (2007) \cite{Scheicher}. According to Scheicher, this algorithm for L\'evy Area has complexity cost of $\mathcal{O}(K)$, where $K$ is the number of partition of the time interval $\Delta{t}$.

    \begin{algorithm}[H]
    \caption{L\'evy Area}\label{algo:Larea}
    \begin{algorithmic}
    \STATE{Define sub-partition length $\Delta^2{t}:=\frac{\Delta{t}}{K}$}
    \STATE{Generate $\mathbf{z}_1,\mathbf{z}_2\sim{}\mathcal{N}_K(\mathbf{0},I_K)$.}
    \STATE{Generate lower triangular matrix of $1$s $T$, set $R:=\Delta^2{t}T$}
    \STATE{Generate lower and upper diagonal matrices of $1$s $L$ and $U$.}
    \STATE{Set $\mathbf{B}_1:=R\mathbf{z}_1$ and $\mathbf{B}_2=:R\mathbf{z}_2$}
    \STATE{$A=\mathbf{b_1}^T(U-L)\mathbf{b_2}$}
    \RETURN{$A$}
    \end{algorithmic}
    \end{algorithm}
    
We may redefine a matrix recursion version of the Milstein Scheme. Consider the following evolutionary dynamic of $\mathbf{S}(t)$:
    \begin{align}
    \mathbf{S}(m+1)=\mathbf{B}(m)\mathbf{S}(m)+\frac{1}{2}\mathbf{b}(m).
    \end{align}
    The matrix $\mathbf{B}(m)$ consists of the first order approximation and the vector $\mathbf{b}(m)$ is the second order approximation. For our SDE system \eqref{FLMMSDEexchange}, $\mathbf{B}(m)$ and $\mathbf{b}(m)$ can be defined as follows:
    \begin{align*}
    &\mathbf{B}(m)=\begin{bmatrix}
    1 + r\Delta{t} + \widetilde{\sigma}_{11}\big(\mathbf{S}(m)\big)\Delta{W}_1(m+1) & \widetilde{\sigma}_{12}\big(\mathbf{S}(m)\big)\Delta{W}_2(m+1)\\
    \widetilde{\sigma}_{21}\big(\mathbf{S}(m)\big)\Delta{W}_1(m+1) & 1 + r\Delta{t} + \widetilde{\sigma}_{21}\big(\mathbf{S}(m)\big)\Delta{W}_2(m+1)
    \end{bmatrix},
    \\
    &\mathbf{b}(m)=\begin{bmatrix} 
    \mathbf{W}^T(m+1)J_1\Sigma\mathbf{W}(m+1)-tr(J_1\Sigma)-\mathbf{1}^T(J_1\Sigma\circ\mathcal{A})\mathbf{1}\\
    \mathbf{W}^T(m+1)J_2\Sigma\mathbf{W}(m+1)-tr(J_2\Sigma)-\mathbf{1}^T(J_2\Sigma\circ\mathcal{A})\mathbf{1}
    \end{bmatrix}.
    \end{align*}
Here $J_i$ is the Jacobi matrix of the $i$-th asset's diffusion functions at the $m$-th step. Matrix $\Sigma$ encapsulates diffusion functions of all assets, also at $m$-th step. They are of the form:
    \begin{gather*}
    J_i=\begin{bmatrix} 
    \frac{\partial{\widetilde{\sigma}}_{i1}}{\partial{s}_1} & \frac{\partial{\widetilde{\sigma}}_{i1}}{\partial{s}_2}\\
    \frac{\partial{\widetilde{\sigma}}_{i2}}{\partial{s}_1} & \frac{\partial{\widetilde{\sigma}}_{i2}}{\partial{s}_2}
    \end{bmatrix}
    ,\quad
    \Sigma=\begin{bmatrix} 
    \widetilde{\sigma}_{11} & \widetilde{\sigma}_{12}\\
    \widetilde{\sigma}_{21} & \widetilde{\sigma}_{22}
    \end{bmatrix}.
    \end{gather*}
$\mathbf{\mathcal{A}}$ is the matrix of L\'evy Areas at step $m$, it has the form:
    \begin{gather*}
    \mathbf{\mathcal{A}}=\begin{bmatrix} 
    0 & \mathcal{A}_{12}\\
    \mathcal{A}_{21} & 0
    \end{bmatrix},
    \end{gather*}
notice $\mathcal{A}$ is an off diagonal matrix, this is because the stochastic integral \eqref{Larea} is $0$ when $i=j$.
\par
It is mentioned in Higham (2015) \cite{Desmond} that Milstein scheme has complexity of $\mathcal{O}(M^2)$ compared to $\mathcal{O}(M)$ of Euler-Maruyama. This is important because Milstein scheme will carry a steeper computation time increase as $M$ increases. 
\subsection{Control Variate Estimator of the Option Price}
\par
The model without liquidity impact is a special case of FLMM. One would naturally assume there exists a high inherited correlation of option prices produced by the two models. It would make sense to use the Magrabe option's value as the control variate of impacted option's value. The Magrabe option can be priced by Magrabe's formula, which uses a pair of correlated GBMs. In fact, we can simultaneously generate the GBM paths while generating FLMM SDEs. We shall do this through Milstein scheme, in the algorithm below; $\mathbf{S}$ and $\mathbf{S}_{cv}$ represents FLMM and GBM asset prices respectively.

\begin{algorithm}[H]
\caption{Milstein Control Variate Path}
\label{algo:cMilstein}
\begin{algorithmic}
\STATE{Initialize Values $\mathbf{S}(t)=\mathbf{S}_{cv}(t)=\mathbf{s}$}
\STATE{Define $\Delta{t}=:\frac{T-t}{M}$}
\FOR{$m=0$ \textbf{to} $M-1$}
\STATE{$\mathbf{\Delta{}W}(m)=\big(\Delta{}w_1(m),\Delta{}w_2(m)\big)\sim\mathcal{N}_2(0,\Delta{t}I_2)$}
\STATE{Set $\mathbf{B}(m)$, $\mathbf{b}(m)$, $J_i$, $\Sigma$ and $\mathbf{\mathcal{A}}$}
\STATE{$\mathbf{S}(m+1)=\mathbf{B}(m)\mathbf{S}(m)+\frac{1}{2}\mathbf{b}(m)$}
\STATE{$\mathbf{S}_{cv}(m+1)=\mathbf{B}_{cv}(m)\mathbf{S}_{cv}(m)+\frac{1}{2}\mathbf{b}_{cv}(m)$}
\ENDFOR
\RETURN{$\mathbf{S}(M)$, $\mathbf{S}_{cv}(M)$}
\end{algorithmic}
\end{algorithm}

\par
By generating $\{\mathbf{S}^{(i)}(M)$, $\mathbf{S}^{(i)}_{cv}(M)\}_{i=1,2,...N}$, we can define the control variate MC estimator of FLMM Exchange Option as follows:
\begin{align}
\label{Estimator}
\overline{V}&=\frac{e^{-r(T-t)}}{N}\sum_{i=1}^N\Big(\big(S_1^{(i)}(M)-S_2^{(i)}(M)\big)^++c\big(S_{cv,1}^{(i)}(M)-S_{cv,2}^{(i)}(M)\big)^+\Big)
\\
&-cV_{Margrabe},\nonumber
\end{align}
here $V_{Margrabe}$ is the price of Magrabe option given by Margrabe's formula in a model without liquidity impact. The term $c$ is the optimization constant. In this case, the variance of our MC estimator is minimized when $\hat{c}=-\frac{Cov(V_{FLMM},V_{Margrabe})}{Var(V_{Margrabe})}$.

\subsection{Option Hedges}
\label{sec:hedge}
Managing the Greeks is an essential part of trading. To determine the Deltas of FLMM Exchange Option, we will adopt the \textit{adjoint method} of Giles and Glasserman (2006) \cite{Glassman}. This method first requires the Greeks to be generated pathwise, then a MC can be applied to estimate the actual value. The adjoint method is advantageous because these pathwise Greeks can be generated simultaneously with the assets. Suppose interchangeability exists between the differential operator and expectation, then the $j$-th Delta of FLMM Exchange Option is:
    \begin{align*}
    \Delta_j(t)=\frac{\partial}{\partial S_j(t)}\widetilde{\mathbb{E}}^{t,s_1,s_2}\Big[e^{-r(T-t)}V\big(\mathbf{S}(T)\big)\Big]=e^{-r(T-t)}\widetilde{\mathbb{E}}^{t,s_1,s_2}\Big[\frac{\partial}{\partial S_j(t)}V\big(\mathbf{S}(T)\big)\Big].
    \end{align*}
By relaxing certain regularity conditions outlined in Glasserman (2004) \cite{glassbook}, we may rewrite it as:
    \begin{align*}
    \frac{\partial}{\partial S_j(t)}V\big(\mathbf{S}(T)\big)=\sum^2_{i=1}\frac{\partial V}{\partial S_i(T)}\frac{\partial S_i(T)}{\partial S_j(t)}.
    \end{align*}
\par
During implementation, $\frac{\partial V}{\partial S_i(T)}$ can be approximated through algorithmic differentiation. While the $\frac{\partial S_i(T)}{\partial S_j(t)}$ term is obtained from taking the path-wise derivative of Milstein scheme \eqref{milsteinscheme}. Set $\Delta_{ij}(t)=\frac{\partial S_i(T)}{\partial S_j(t)}$, we obtain an approximating scheme for $\Delta_{ij}(m)$ as follows:
    \begin{align*}
    \Delta_{ij}(m+1)&=\Delta_{ij}(m)+r\Delta_{ij}(m)\Delta{t}+\sum_{k,l=1}^2\frac{\partial \widetilde{\sigma}_{ik}}{\partial s_l}\Delta_{lj}(m)\Delta{W}_k(m+1)
    \\
    &+\frac{1}{2}\sum_{k,l,p,q=1}^2\Delta_{qj}(m)\Big(\frac{\partial^2\widetilde{\sigma}_{ik}}{\partial s_p\partial s_q}\widetilde{\sigma}_{pj}\big(\mathbf{S}(m)\big)+\frac{\partial \widetilde{\sigma}_{ik}}{\partial s_p}\frac{\partial \widetilde{\sigma}_{pl}}{\partial s_q}\Big),
    \end{align*}
where $m=0,1,...M-1$. If we define a matrix $\mathbf{D}(m)$ as:
    \begin{align*}
    D_{ij}(m)&=\delta_{ij}(m)+r\Delta{t}+\sum_{k=1}^2\frac{\partial \widetilde{\sigma}_{ik}}{\partial s_j}\Delta{W}_{k}(m+1)
    \\
    &+\frac{1}{2}\sum_{k,l,p=1}^2\Big(\frac{\partial^2\widetilde{\sigma}_{ik}}{\partial s_p\partial s_j}\widetilde{\sigma}_{pj}\big(\mathbf{S}(m)\big)+\frac{\partial \widetilde{\sigma}_{ik}}{\partial s_p}\frac{\partial \widetilde{\sigma}_{pl}}{\partial s_j}\Big),
    \end{align*}
then the evolution of $\mathbf{\Delta}$ can be redefined using matrix recursion as follows:
    \begin{align*}
    \mathbf{\Delta}(m+1)=D(m)\mathbf{\Delta}(m),
    \end{align*}
where $\mathbf{\Delta}(t)=I$. Similar to estimating the option price, we a can use the Delta from the Magrabe option as a multivariate control variate. We adopt the method presented by Rubinstein and Marcus (1985) \cite{Rubin} and set up the estimator for Delta:
    \begin{align}
    \label{EstimatorDelta}
    \overline{\mathbf{\Delta}}&=\frac{e^{-r(T-t)}}{N}\sum_{i=1}^N\Big(\mathbf{\Delta}^{(i)}(M)+C_1\mathbf{\Delta}_{cv}^{(i)}(M)\Big)-C_1\mathbf{\Delta}_{Margrabe}.
    \end{align}
The variance of $\overline{\mathbf{\Delta}}$ is minimized when $\hat{C}_1=\Sigma_{\mathbf{\Delta}\mathbf{\Delta}_{cv}}\Sigma_{\mathbf{\Delta}_{cv}\mathbf{\Delta}_{cv}}^{-1}$.

\subsection{Experimental results}
\label{sec:experiments}
We implement our MC engine with alternating space and path parameter for the purpose of determining the effect on a $99\%$ Gaussian confidence interval (CI). For consistency, we fix a set of option parameters: $s_1=60$, $s_2=80$, $T=0.5$, $t=0$, $\sigma_1=0.4$, $\sigma_2=0.2$, $\rho=0.5$ and $r=0.05$. We also fix the price impact function parameters to: $\epsilon=0.04$ and $\beta=100$. The numerical results are presented below:
\begin{table}[H]
\caption{Space (N) Dimension MC Results}
\centering
\begin{tabular}{|c|c|c|c|c|c|c|} 
\hline
\textbf{N}&\textbf{M}& \textbf{$\overline{V}$} &\textbf{$99\%$ CI of $\overline{V}$}& \textbf{CI Length} & \textbf{CPU Time}
\\
\hline
$100$ & $100$ & $1.0008$ & $[0.998642,1.00295]$ & $0.0043124$ & 0.11s\\
\hline
$1000$ & $100$ & $1.00145$ & $[1.00088,1.00201]$ & $0.00112514$ & 1.08s\\
\hline
$10k$ & $100$ & $1.0013$ & $ [1.0011,1.00151] $ & $0.000412377$ & 9.84s \\
\hline
$100k$ & $100$ & $1.00134$ & $[1.00128,1.0014]$ & $0.000126287$ & 99.97s\\
\hline
$1m$ & $100$ & 1.00139 & $[1.00137,1.00141]$ & $0.0000405683$ & 1033.30s \\
\hline
\end{tabular}
\label{tab:miltable1}
\end{table}
\begin{table}[H]
\caption{Path (M) Dimension MC Results}
\centering
\begin{tabular}{|c|c|c|c|c|c|c|} 
\hline
\textbf{N}&\textbf{M}& \textbf{$\overline{V}$} &\textbf{$99\%$ CI of $\overline{V}$}& \textbf{CI Length} & \textbf{CPU Time}
\\
\hline
$1000$ & $100$ & $1.00145$ & $[1.00088,1.00201]$ & $0.00112514$ & 1.08s\\
\hline
$1000$ & $200$ & $1.00129$ & $[1.00072,1.00186]$ & $0.00113221$ & 3.82s\\
\hline
$1000$ & $400$ & $1.00181$ & $ [1.00118,1.00244] $ & $0.00126276$ & 15.63s \\
\hline
$1000$ & $800$ & $1.00111$ & $[1.00053,1.00169]$ & $0.00115177$ & 66.96s\\
\hline
$1000$ & $1600$ & $1.00151$ & $[1.00088,1.00215]$ & $0.00127271$ & 237.92s \\
\hline
\end{tabular}
\label{tab:miltable2}
\end{table}
One observation from our experiment is that as the path dimension doubles, the computation time almost quadruples. This is in agreement with Higham's assertion on the complexity cost of Milstein Scheme.
\begin{figure}[H]
\centering
\begin{subfigure}[b]{0.475\textwidth}
\includegraphics[height=4.19cm]{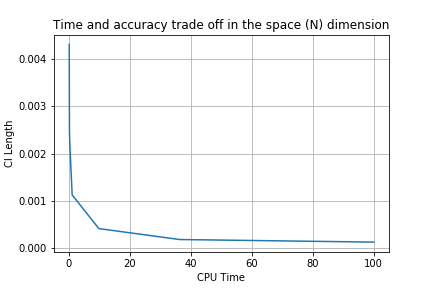}
\label{graph:Ndim}
\end{subfigure}\quad
\begin{subfigure}[b]{0.475\textwidth}
\includegraphics[height=4.19cm]{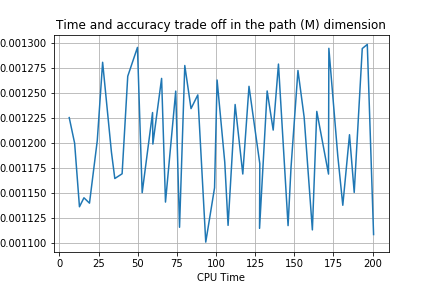}
\label{graph:Mdim}
\end{subfigure}
\end{figure}
From a practitioner's point of view, we must consider the trade off between time complexity and accuracy of estimation. In the first graph above, we observe that as the computation time increases in space dimension, the length of CI exponentially decays. However, as we increase the computation time in path dimension, there is an ambiguous effect on CI length. This is emphasized in the second graph above. It is fairly self-explanatory that we should focus our computation resources on the space dimension to get the best complexity vs accuracy trade off.
\par
We also would like to compare FLMM against the frictionless model. In particular, we want to confirm the liquidity impact of FLMM requires a strictly positive valuation adjustment.
\begin{table}[H]
\caption{Analysis of Liquidity Premium}
\centering
\begin{tabular}{ccccc} 
\hline
$N=100000$, $M=100$&\textbf{$s_1$}& \textbf{FLMM} &\textbf{Margrabe}& \textbf{Excess Price}
\\
\hline
$s_2=10$ & $10$ & $0.98591$ & $0.974767$ & $0.011143$\\
& $20$ & $0.00237199$ & $0.00236962$ & $0.00000236655$ \\
\hline
& $10$ & $0.00237514$ & $0.00236962$ & $0.00000552152$ \\
$s_2=20$ & $20$ & $1.96065$ & $1.94953$ & $0.0111181$\\
& $30$ & $0.122669$ & $0.121575$ & $0.0010937$\\
\hline
& $20$ & $0.122554$ & $0.121575$ & $0.000978586$ \\
$s_2=30$ & $30$ & $2.93598$ & $2.9243$ & $0.0116819$\\
& $40$ & $10.504$ & $10.499$ & $0.00496871$\\
\hline
& & \vdots & & \\
\hline
& $90$ & $5.10866$ & $5.09879$ & $0.00987472$\\
$s_2=100$ & $100$ & $9.75846$ & $9.74767$ & $0.010783$\\
\hline
\end{tabular}
%\footnotesize{Parameters: $T=0.5$, $t=0$, $\sigma_1=0.4$, $\sigma_2=0.2$, $\rho=0.5$ and $r=0.05$}
\label{tab:miltable3}
\end{table}
From our experiments, we indeed observe an excess in the option price due to the FLMM. This premium seems to be the greatest for at-the-money options. Furthermore, as the trade-cost-per-share parameter $\epsilon$ increases, we observe a higher liquidity premium. This effect is illustrated in the figures below.
\begin{figure}[H]
\caption{Liquidity Value Adjustment}
\centering
\includegraphics[width=8cm]{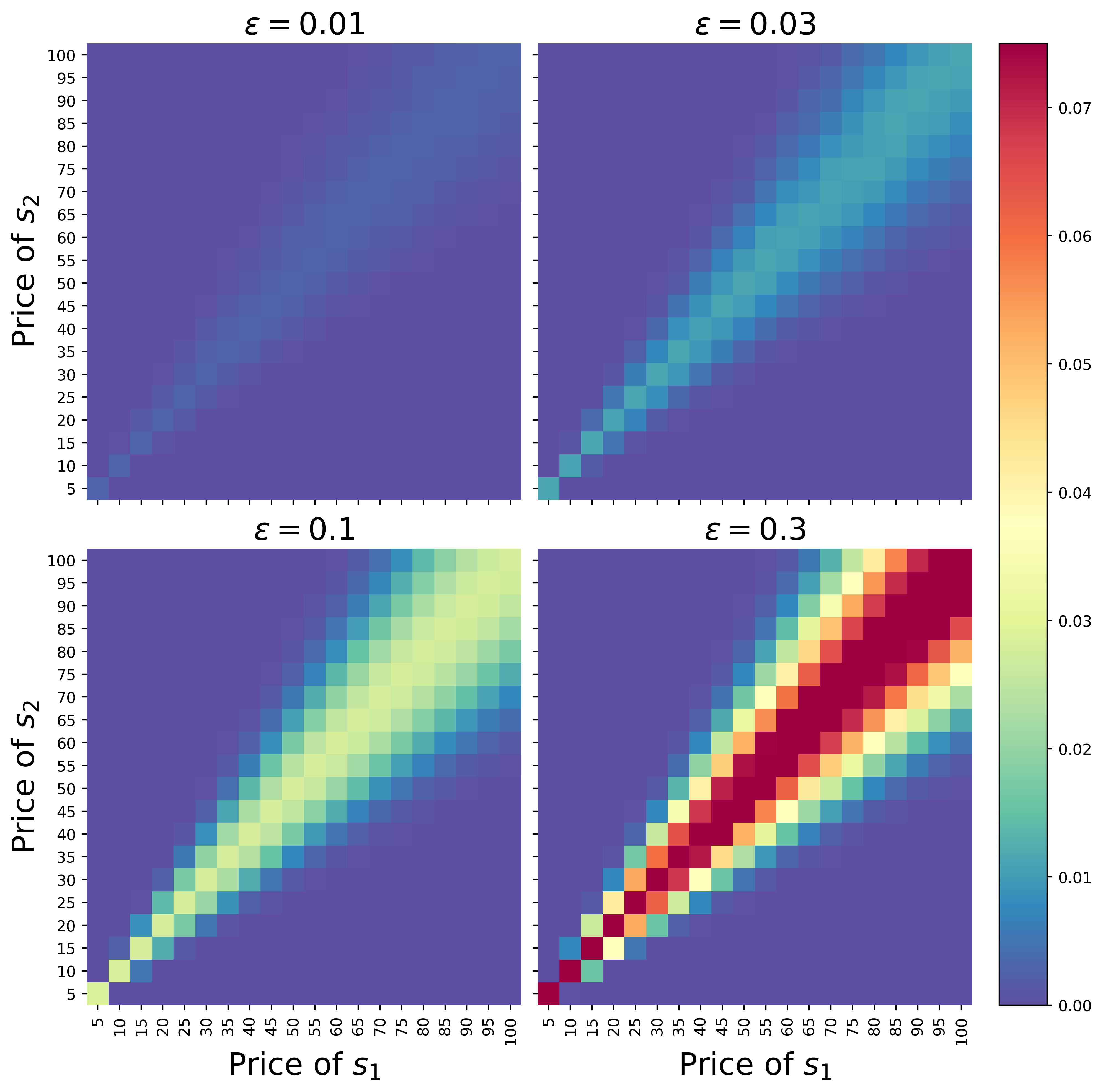}
\label{fig:excess}
\end{figure}
It only appears natural to be also interested in the liquidity adjustment for Delta. Using the Margrabe Delta as a reference, one would expect that the strictly greater price of our illiquid asset $1$ would cause $\Delta_1$ to be greater and $\Delta_2$ to be less. Empirically, we observe an excess effect in $\Delta_1$, but we also observed an excess effect in $\Delta_2$. We illustrate this surprising result in the figures below.
\begin{figure}[H]
\caption{Liquidity Delta Adjustment}
\centering
\begin{subfigure}[b]{0.475\textwidth}
\includegraphics[height=6cm]{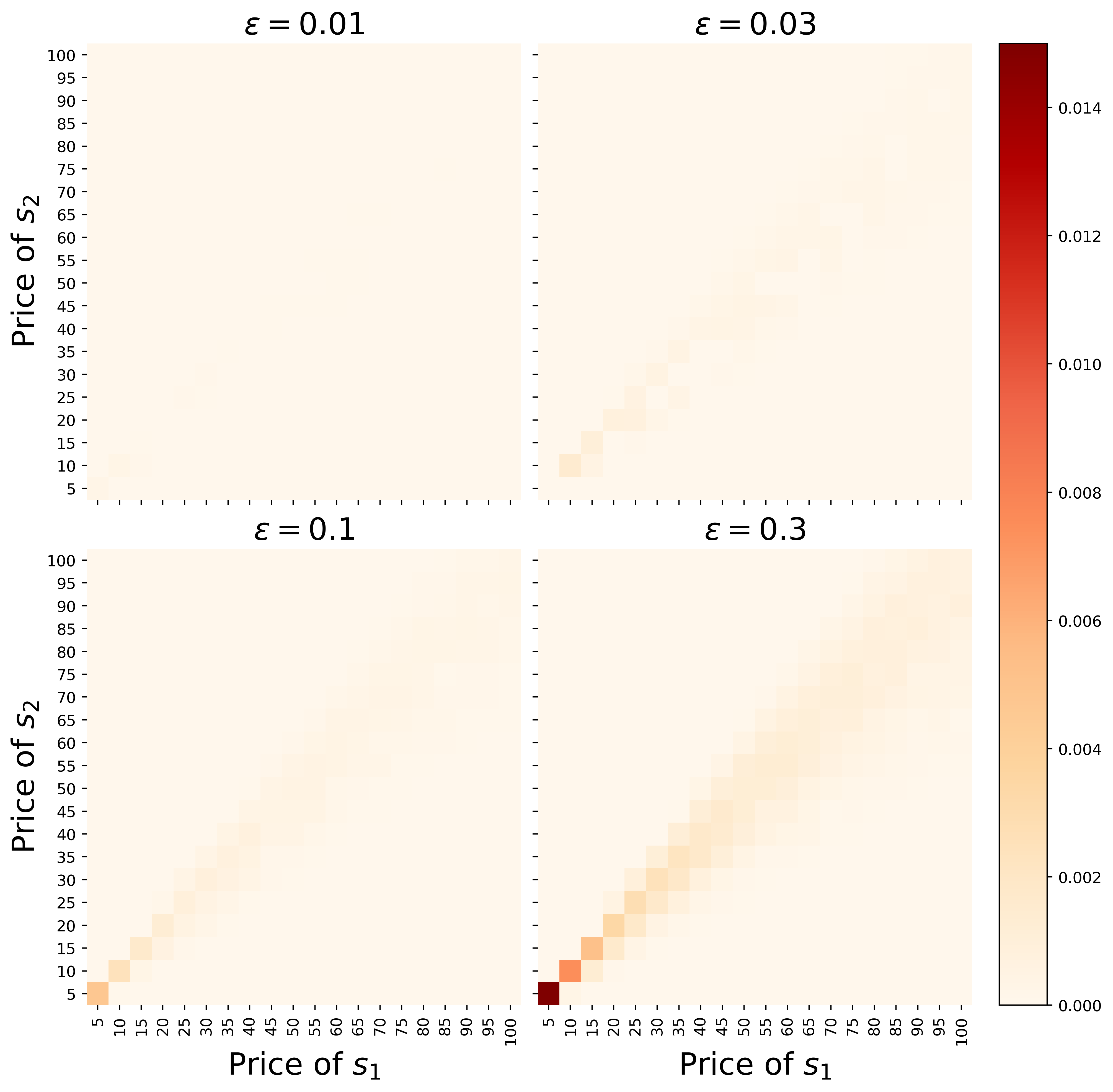}
\label{graph:HeatD1}
\end{subfigure}\quad
\begin{subfigure}[b]{0.475\textwidth}
\includegraphics[height=6cm]{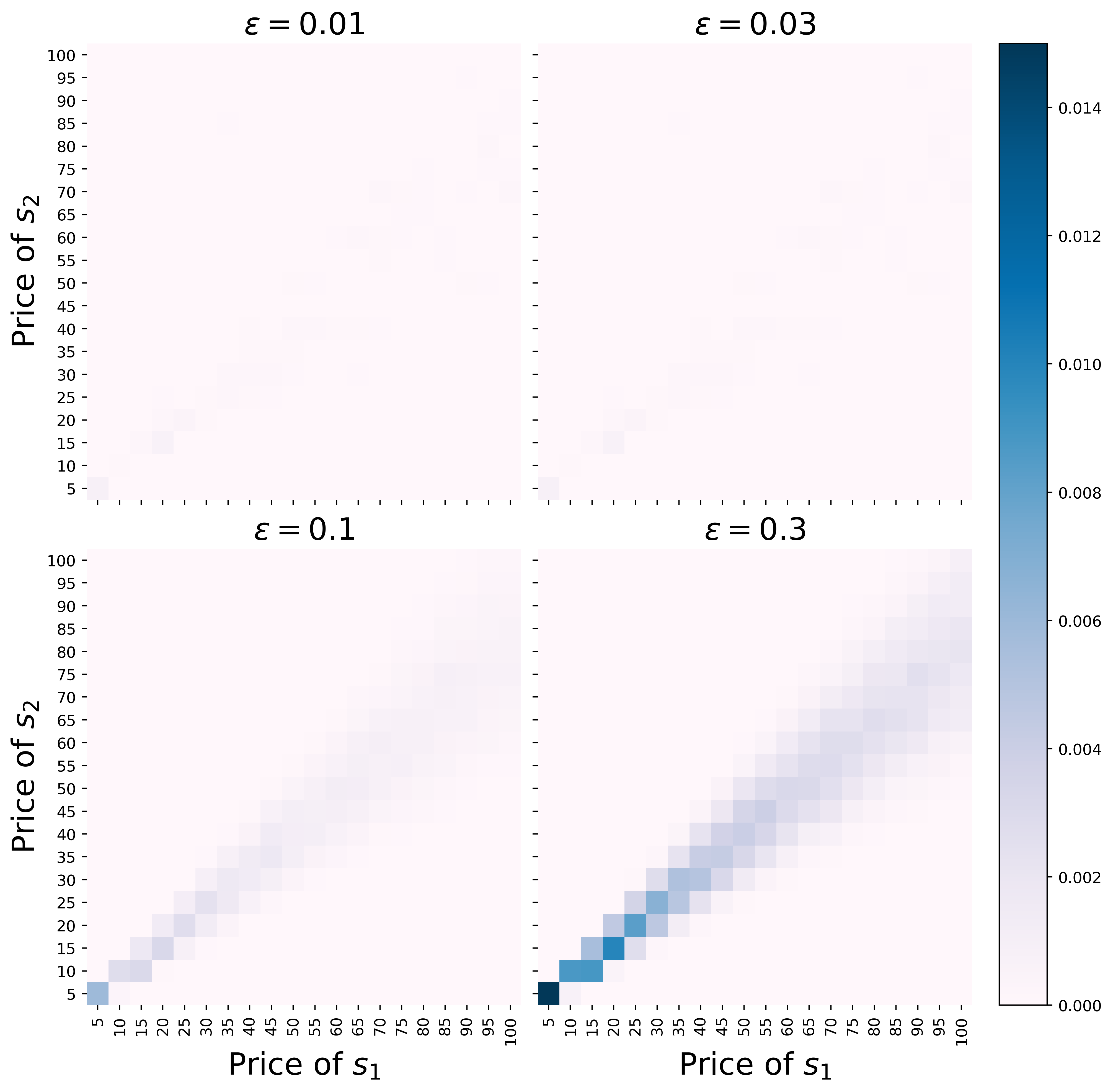}
\label{graph:HeatD2}
\end{subfigure}
\end{figure}
These positive Delta adjustments effects also reach their respective pinnacles when the option is at-the-money. In a model with transaction costs these
Delta adjustments would add extra value to the option price. 
\section{Deep Learning Method}
\label{sec:deep}
\par
Artificial neural network have powerful predictive capabilities, one of the first versions are the FFN. This network is structured as a sequence of layers, with various numbers of neurons embedded in each layer. We shall use $N$ to denote the number of layers, and $n_i$ to denote the number of neurons in the $i$-th layer. In a fully connected FFN, each neuron in the current layer has a connection with each neuron in the subsequent layer. The strength of these connections are known as \textit{weights}, we denote the weights connected to the $j$-th neuron in the $i$-th layer as $\mathbf{w}^{[i]}_j$. Each neuron also carries a unique bias term $b^{[i]}_j$, this term has a similar effect as the regression intercept. The final component of a neuron is the activation function $f(z)$, similar to linking functions of non-linear regression, its purpose is to add non-linearity. In this study, we used these types of activation functions:
\begin{table}[H]
\caption{Activation Functions}
\label{active}
\centering
\begin{tabular}{|c|c|} 
\hline
Type & Activation Function\\
\hline
ReLU & $f(z)=\max(z,0)$\\
\hline
SoftPlus & $f(z)=\log(1+e^z)$\\
\hline
\end{tabular}
\end{table}
The operation of a neuron can be expressed as:
\begin{gather*}
z^{[i]}_j=\mathbf{w}^{[i]}_j\mathbf{h}^{[i-1]}+b^{[i]}_j,
\\
h^{[i-1]}_j=f(z^{[i]}_j).
\end{gather*}
We also provide a computation graph on the $j$-th neuron in the $i$-th layer:
\begin{figure}[H]
\centering
\resizebox{5cm}{5cm}{%
\begin{tikzpicture}[
innode/.style={circle, draw=yellow, fill=yellow!20, very thick, minimum size=2mm},
node/.style={circle, draw=yellow, fill=yellow!20, very thick, minimum size=2mm},
textnode/.style={rectangle, draw=white, fill=white!20, very thick, minimum size=5mm},
]

%Neuron
\node[innode] ($y$) at (-3.5,3) {$h_1^{[i-1]}$};
\node[innode] ($y$) at (-3.5,1.5) {$h_2^{[i-1]}$};
\node[innode] ($y$) at (-3.5,0) {$h_3^{[i-1]}$};
\node[innode] ($y$) at (-3.5,-3) {$h_n^{[i-1]}$};
\node[textnode] ($y$) at (-2,2.5) {$w^{[i]}_{i1}$};
\node[textnode] ($y$) at (-2,1.3) {$w^{[i]}_{i2}$};
\node[textnode] ($y$) at (-2,0.3) {$w^{[i]}_{i3}$};
\node[textnode] ($y$) at (-2,-1.25) {$w^{[i]}_{in_i}$};
\draw[thick,->] (-2.8,3) -- (-0.75,0.2);
\draw[thick,->] (-2.8,1.5) -- (-0.75,0.1);
\draw[thick,->] (-2.8,0) -- (-0.75,0);
\draw[thick,->] (-2.8,-3) -- (-0.75,-0.2);
\draw[decorate sep={0.4mm}{1.5mm},fill] (-3.5,-2.2) -- (-3.5,-0.75);
\node[node] ($y$) at (0,0) {$f\big(z^{[i]}_j\big)$};
\node[textnode] ($y$) at (1.5,0.3) {$h^{[i]}_j$};
\draw[thick,->] (0.75,0) -- (3,0);

\end{tikzpicture}
}
\end{figure}
This process is repeated for every single neuron, which allows us to transverse through the network and arrive at the output layer $h^{[N]}=\hat{y}$ (For the purpose of option pricing, we have a single output $h^{[N]}$, but in general $h^{[N]}$ is a vector). This entire process is often referred to as \textit{forward propagation}. The figure below describes the FFN architecture deployed to price Exchange Option under FLMM:
\begin{figure}[H]
\centering
\resizebox{9.36cm}{6cm}{%
\begin{tikzpicture}[
innode/.style={circle, draw=red, fill=red!20, very thick, minimum size=7mm},
hidnode/.style={circle, draw=yellow, fill=yellow!20, very thick, minimum size=9mm},
outnode/.style={circle, draw=green, fill=green!20, very thick, minimum size=7mm},
textnode/.style={rectangle, draw=white, fill=white!20, very thick, minimum size=5mm},
]

%texts
\node[textnode] at (-4.75,-5.9) {\tiny{Input Layer}};
\node[textnode] at (-2.75,-5.5) {\tiny{Hidden Layer 1}};
\node[textnode] at (-0.5,-5.5) {\tiny{Hidden Layer 2}};
\node[textnode] at (1.8,-5.5) {\tiny{Hidden Layer 3}};
\node[textnode] at (4.05,-5.5) {\tiny{Hidden Layer 4}};
\node[textnode] at (6.05,-2.85) {\tiny{Output Neuron}};

%Input Layer Nodes
\node[innode] ($x_1$) at (-4.75, 1) {$s_1$};
\node[innode] ($x_1$) at (-4.75, 0) {$s_2$};
\node[innode] ($x_2$) at (-4.75,-1) {$\sigma_1$};
\node[innode] ($x_n$) at (-4.75,-2) {$\sigma_2$};
\node[innode] ($x_n$) at (-4.75,-3) {$r$};
\node[innode] ($x_n$) at (-4.75,-4) {$\rho$};
\node[innode] ($x_n$) at (-4.75,-5) {$\tau$};
\draw[blue, very thick] (-5.25,-5.5) rectangle (-4.25,1.5);

%Lines Input to 1st Hidden
\draw[->] (-4.35, 1) -- (-3.25,0.5);
\draw[->] (-4.35, 1) -- (-3.25,-0.75);
\draw[->] (-4.35, 1) -- (-3.25,-2.0);
\draw[->] (-4.35, 1) -- (-3.25,-4.5);
\draw[->] (-4.35, 0) -- (-3.25,0.5);
\draw[->] (-4.35, 0) -- (-3.25,-0.75);
\draw[->] (-4.35, 0) -- (-3.25,-2.0);
\draw[->] (-4.35, 0) -- (-3.25,-4.5);
\draw[->] (-4.35, -1) -- (-3.25,0.5);
\draw[->] (-4.35, -1) -- (-3.25,-0.75);
\draw[->] (-4.35, -1) -- (-3.25,-2.0);
\draw[->] (-4.35, -1) -- (-3.25,-4.5);
\draw[->] (-4.35, -2) -- (-3.25,0.5);
\draw[->] (-4.35, -2) -- (-3.25,-0.75);
\draw[->] (-4.35, -2) -- (-3.25,-2.0);
\draw[->] (-4.35, -2) -- (-3.25,-4.5);
\draw[->] (-4.35, -3) -- (-3.25,0.5);
\draw[->] (-4.35, -3) -- (-3.25,-0.75);
\draw[->] (-4.35, -3) -- (-3.25,-2.0);
\draw[->] (-4.35, -3) -- (-3.25,-4.5);
\draw[->] (-4.35, -4) -- (-3.25,0.5);
\draw[->] (-4.35, -4) -- (-3.25,-0.75);
\draw[->] (-4.35, -4) -- (-3.25,-2.0);
\draw[->] (-4.35, -4) -- (-3.25,-4.5);
\draw[->] (-4.35, -5) -- (-3.25,0.5);
\draw[->] (-4.35, -5) -- (-3.25,-0.75);
\draw[->] (-4.35, -5) -- (-3.25,-2.0);
\draw[->] (-4.35, -5) -- (-3.25,-4.5);

%1st Hidden Layer Nodes
\node[hidnode] ($x_1$) at (-2.75, 0.5) {$h_{1}^{[1]}$};
\node[hidnode] ($x_2$) at (-2.75,-0.75) {$h_{2}^{[1]}$};
\node[hidnode] ($x_n$) at (-2.75,-2.0) {$h_{3}^{[1]}$};
\draw[decorate sep={0.4mm}{1.5mm},fill] (-2.75, -2.75) -- (-2.75,-3.75);
\node[hidnode] ($x_n$) at (-2.75,-4.5) {$h_{n}^{[1]}$};
\draw[blue, very thick] (-3.35,-5.1) rectangle (-2.15,1.1);

%Lines 1st Hidden to 2nd Hidden
\draw[->] (-2.25,0.5) -- (-1,0.5);
\draw[->] (-2.25,0.5) -- (-1,-0.75);
\draw[->] (-2.25,0.5) -- (-1,-2.0);
\draw[->] (-2.25,0.5) -- (-1,-4.5);
\draw[->] (-2.25,-0.75) -- (-1,0.5);
\draw[->] (-2.25,-0.75) -- (-1,-0.75);
\draw[->] (-2.25,-0.75) -- (-1,-2.0);
\draw[->] (-2.25,-0.75) -- (-1,-4.5);
\draw[->] (-2.25,-2) -- (-1,0.5);
\draw[->] (-2.25,-2) -- (-1,-0.75);
\draw[->] (-2.25,-2) -- (-1,-2.0);
\draw[->] (-2.25,-2) -- (-1,-4.5);
\draw[->] (-2.25,-4.5) -- (-1,0.5);
\draw[->] (-2.25,-4.5) -- (-1,-0.75);
\draw[->] (-2.25,-4.5) -- (-1,-2.0);
\draw[->] (-2.25,-4.5) -- (-1,-4.5);

%2nd Hidden Layer Nodes
\node[hidnode] ($x_1$) at (-0.5, 0.5) {$h_{1}^{[2]}$};
\node[hidnode] ($x_2$) at (-0.5,-0.75) {$h_{2}^{[2]}$};
\node[hidnode] ($x_n$) at (-0.5,-2.0) {$h_{3}^{[2]}$};
\draw[decorate sep={0.4mm}{1.5mm},fill] (-0.5, -2.75) -- (-0.5,-3.75);
\node[hidnode] ($x_n$) at (-0.5,-4.5) {$h_{n}^{[2]}$};
\draw[blue, very thick] (-1.10,-5.1) rectangle (0.1,1.1);

%Lines 2nd Hidden to 3rd Hidden
\draw[->] (0,0.5) -- (1.25,0.5);
\draw[->] (0,0.5) -- (1.25,-0.75);
\draw[->] (0,0.5) -- (1.25,-2.0);
\draw[->] (0,0.5) -- (1.25,-4.5);
\draw[->] (0,-0.75) -- (1.25,0.5);
\draw[->] (0,-0.75) -- (1.25,-0.75);
\draw[->] (0,-0.75) -- (1.25,-2.0);
\draw[->] (0,-0.75) -- (1.25,-4.5);
\draw[->] (0,-2) -- (1.25,0.5);
\draw[->] (0,-2) -- (1.25,-0.75);
\draw[->] (0,-2) -- (1.25,-2.0);
\draw[->] (0,-2) -- (1.25,-4.5);
\draw[->] (0,-4.5) -- (1.25,0.5);
\draw[->] (0,-4.5) -- (1.25,-0.75);
\draw[->] (0,-4.5) -- (1.25,-2.0);
\draw[->] (0,-4.5) -- (1.25,-4.5);

%3rd Hidden Layer Nodes
\node[hidnode] ($x_1$) at (1.75, 0.5) {$h_{1}^{[3]}$};
\node[hidnode] ($x_2$) at (1.75,-0.75) {$h_{2}^{[3]}$};
\node[hidnode] ($x_n$) at (1.75,-2.0) {$h_{3}^{[3]}$};
\draw[decorate sep={0.4mm}{1.5mm},fill] (1.75, -2.75) -- (1.75,-3.75);
\node[hidnode] ($x_n$) at (1.75,-4.5) {$h_{n}^{[3]}$};
\draw[blue, very thick] (1.15,-5.1) rectangle (2.35, 1.1);

%Lines 3rd Hidden to 4th Hidden
\draw[->] (2.25,0.5) -- (3.5,0.5);
\draw[->] (2.25,0.5) -- (3.5,-0.75);
\draw[->] (2.25,0.5) -- (3.5,-2.0);
\draw[->] (2.25,0.5) -- (3.5,-4.5);
\draw[->] (2.25,-0.75) -- (3.5,0.5);
\draw[->] (2.25,-0.75) -- (3.5,-0.75);
\draw[->] (2.25,-0.75) -- (3.5,-2.0);
\draw[->] (2.25,-0.75) -- (3.5,-4.5);
\draw[->] (2.25,-2) -- (3.5,0.5);
\draw[->] (2.25,-2) -- (3.5,-0.75);
\draw[->] (2.25,-2) -- (3.5,-2.0);
\draw[->] (2.25,-2) -- (3.5,-4.5);
\draw[->] (2.25,-4.5) -- (3.5,0.5);
\draw[->] (2.25,-4.5) -- (3.5,-0.75);
\draw[->] (2.25,-4.5) -- (3.5,-2.0);
\draw[->] (2.25,-4.5) -- (3.5,-4.5);

%4th Hidden Layer Nodes
\node[hidnode] ($x_1$) at (4, 0.5) {$h_{1}^{[4]}$};
\node[hidnode] ($x_2$) at (4, -0.75) {$h_{2}^{[4]}$};
\node[hidnode] ($x_n$) at (4, -2.0) {$h_{3}^{[4]}$};
\draw[decorate sep={0.4mm}{1.5mm},fill] (4, -2.75) -- (4,-3.75);
\node[hidnode] ($x_n$) at (4,-4.5) {$h_{n}^{[4]}$};
\draw[blue, very thick] (3.4,-5.1) rectangle (4.6, 1.1);

%Lines 4th Hidden to Output 
\draw[->] (4.5,0.5) -- (5.6,-2);
\draw[->] (4.5,-0.75) -- (5.6,-2);
\draw[->] (4.5,-2) -- (5.6,-2);
\draw[->] (4.5,-4.5) -- (5.6,-2);

%Output Layer
\node[outnode] ($y$) at (6,-2) {$V$};
\draw[blue, very thick] (5.5, -1.55) rectangle (6.5, -2.45);

%4th Hidden to Output Lines
\end{tikzpicture}
}
\end{figure}
\par
The \textit{loss function} measures the goodness of fit. We use mean squared error (MSE) as the loss function, which is commonly used in regression analysis. We will use MSE to evaluate the result of the forward propagation. This evaluation is preformed for every $B$ input, $B$ is known as the \textit{batch size}. Our loss function is formulated as:
    \begin{align*}
    \mathcal{L}(\mathbf{\hat{y}},\mathbf{y})=\sum_{k=1}^{B}(\hat{y}_k-y_k)^2.
    \end{align*}
Minimization of the loss function follows the steepest descent idea, so one has to compute gradient fields with respect to the weights and biases. This is often accomplished through algorithmic differentiation referenced as \textit{back propagation}. Then, the weights and biases are updated in the direction of the gradient field, in hope of discovering a ``good enough" local minimum. The common choice of methodology for optimization is the \textit{batch gradient descent} method. This method is demonstrated as: 
    \begin{gather*}
    \mathbf{w}_j^{[i],(new)}=\mathbf{w}^{[i],(old)}_j-\alpha\frac{\partial \mathcal{L}}{\partial \mathbf{w}_j^{[i],(old)}},
    \\
    \mathbf{b}_j^{[i],(new)}=\mathbf{b}_j^{[i],(old)}-\alpha\frac{\partial \mathcal{L}}{\partial \mathbf{b}_j^{[i],(old)}},
    \\
    \text{for}\quad{}j=1,2,...n_i,\text{ and }i=1,2,...N.
    \end{gather*}
In the above expression, $\alpha$ is the \textit{learning rate}. 
\par
One batch of forward propagation combined with one instance of back propagation is considered as one iteration of batch training. An \textit{epoch} encompasses a series of batch training that exhausts the entire data set. Normally, the training is either repeated for a fixed number of epochs, or stopped early when the loss function ceases to decrease further.
\par
The central theorem in neural networks is the \textit{universal approximation theorem}. This theorem highlights the approximation power of FFNs. Hornik (1989) \cite{HORNIK} established the fact that deep FFNs are universal approximators, in other words, any function can be accurately approximated by some deep FFN. Since option prices are smooth solutions of PDEs, then it should be feasible to predict these solutions with FFNs.
\subsection{Deeply Learning Derivative}
\par
Option pricing can often be computationally expensive. Ferguson and Green (2018) \cite{Ferguson} demonstrated the power of FFN, and achieved a much faster speed than traditional MC engines when pricing baskets. However, the initial costs comes from generation of option inputs, as well as, estimating the corresponding option values through MC engines. Furthermore, training and calibrating the FFN takes tedious effort as well. Nevertheless, these ``costs" are reasonable to large financial institutions, and at least in theory, will integrate well with their operations. This is largely because both the data generation and network training can be done offline, when the markets are closed. In addition, the input space can be restricted to reflect a set of likely market scenarios.
\par
To build a FFN pricer for our FLMM Exchange Option, we will use Algorithm \eqref{Estimator} as the underlying MC engine. Our estimator has 7 parameters
\\
\big($\mathbf{x}=(s_1,s_2,r,\rho,\sigma_1,\sigma_2,\tau)$\big), a set of these parameters count as $1$ sample input. It is important to emphasize the particular distribution used to generate the inputs, these should be unique for each option. Indeed, some factors to be considered when choosing the distributions are:
\begin{itemize}
  \item The physical meaning of each underlying parameter.
  \item The payoff function itself should be considered because it is pointless to generate excessive of out-of-money MC paths.
\end{itemize}
\par
Generating the inputs in judicious ways will not only help the loss to converge faster, but will also help the FFN to approximate a meaningful solution.
In our case, we adopted an even spilt between $2$ data generation schemes. The first method allows us to sample unbiasly from the entire input space. The second method will allow us to sample more realistic input parameters, as well as, capture more in-the-money payout paths.
\begin{table}[H]
\caption{Data Generation Schemes}
\centering
\begin{tabular}{|c|c|c|}
\hline
Parameter & Method 1 & Method 2\\
\hline
$s_1$ & $s_1\sim\mathcal{U}(0,100)$ & $s_1\sim50\exp(X_1)$, $X_1\sim\mathcal{N}(0.5,0.25)$\\
\hline
$s_2$ & $s_2\sim\mathcal{U}(0,100)$ & $s_2\sim50\exp(X_1-X_2)$, $X_2\sim\mathcal{N}(0.5,0.25)$\\
\hline
$\sigma_1$ & $\sigma_1\sim\mathcal{U}(0,0.5)$ & $\sigma_1\sim\mathcal{U}(0,0.5)$\\
\hline
$\sigma_2$ & $\sigma_2\sim\mathcal{U}(0,0.5)$ & $\sigma_2\sim\mathcal{U}(0,0.5)$\\
\hline
$r$ & $r\sim\mathcal{U}(0,0.1)$ & $r\sim\mathcal{U}(0,0.1)$\\
\hline
$\rho$ & $\rho\sim\mathcal{U}(1,-1)$ & $\rho\sim{}2(X_3-0.5)$, $X_3\sim\mathcal{\beta}(5,2)$\\
\hline
$\tau$ & $\tau\sim\mathcal{U}(0,2)$ & $\tau\sim\mathcal{U}(0,2)$\\
\hline
\end{tabular}
\end{table}
The implementation of \textit{Deeply Learning Derivative} method can be synthesized by the following programming architectural graph:

\begin{figure}[H]
\centering
\resizebox{7cm}{7cm}{%
\begin{tikzpicture}[
node/.style={rectangle, draw=black, fill=blue!20, very thick, minimum size=5mm},
textnode/.style={rectangle, draw=gray!40, fill=gray!40, very thick, minimum size=5mm},
outnode/.style={rectangle, draw=black, fill=green!20, very thick, minimum size=5mm},
]
\draw[black, very thick,fill=gray!40] (-4,0) rectangle (-1,2);
\node[textnode] ($x_1$) at (-2.5, 1.5) {Data Generation};
\node[node] ($x_1$) at (-2.5, 1) {Method 1};
\node[node] ($x_1$) at (-2.5, 0.5) {Method 2};
\draw[->,very thick,line width=0.75mm] (-2.5, 0) -- (0.75, -1.95);
\draw[->,very thick,line width=0.75mm] (-1, 1) -- (3, 1);

\draw[black, very thick,fill=gray!40] (3,0) rectangle (6,2);
\node[textnode] ($x_1$) at (4.5, 1.5) {MC Engine};
\node[node] ($x_1$) at (4.5, 1) {Milstein};
\node[node] ($x_1$) at (4.5, 0.5) {Control Variate};
\draw[->,very thick,line width=0.75mm] (4.5, 0) -- (1.25, -1.95);

\draw[black, very thick,fill=gray!40] (-1.5,-2) rectangle (3.5,-4.6);
\node[textnode] ($x_1$) at (1, -2.5) {Deep Feed Forward Network};
\node[node] ($x_1$) at (1, -3) {Training};
\node[node] ($x_1$) at (1, -3.55) {Cross-Validation};
\node[node] ($x_1$) at (1, -4.1) {Testing};

\draw[->,very thick,line width=0.75mm] (1, -4.6) -- (1, -6.5);
\draw[black, very thick,fill=gray!40] (-1.25, -6.5) rectangle (3.25,-7.5);
\node[outnode] ($x_1$) at (1, -7) {Export to Production};

\end{tikzpicture}
}
\end{figure}
\subsection{Experimental Results}
\par
The FFN contains 4 fully connected deep layers with $300$ ReLu neurons per layer. The output layer contain a single SoftPlus Neuron to ensure the prediction would be positive definite. We generated $1$ million inputs, and uses a relatively inaccurate MC engine ($N$=100,$M$=100) to construct the training set. The logic is it has been shown in practice a well-trained deep FFN has the ability to remove the inaccuracy of weak MC estimators. We trained the FNN with mini-batch size of $1024$, and updated the gradient with ADAM optimizer (2015) \cite{Kingma}. We performed validation with samples created from a highly accurate MC engine ($N$=100k,$M$=100), at a $100/1$ ratio. Initially, the FFN was set to train for $1000$ epochs. After $850$ epochs of training, the loss function cease to decreases further significantly. To prevent over-fitting, it is justifiable to apply early stopping.
\begin{figure}[H]
\centering
\begin{subfigure}[b]{0.475\textwidth}
\includegraphics[width=6.2cm]{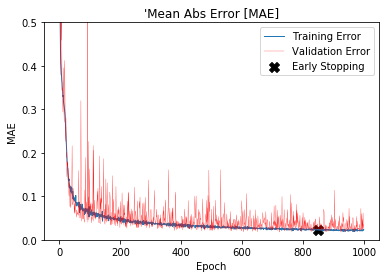}
\end{subfigure}\quad
\begin{subfigure}[b]{0.475\textwidth}
\includegraphics[width=6.2cm]{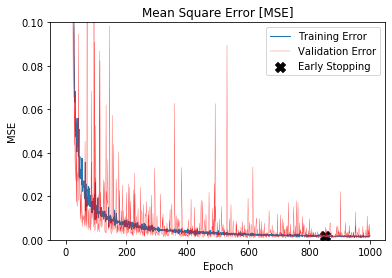}
\end{subfigure}
\end{figure}
\par
We observe both of the mean absolute error (MAE) and MSE of the validation set oscillate around the training set. Furthermore, the amplitude of the oscillation decreases as we train our network. This implies our network is learning to minimize in terms of $\mathcal{L}_1$ and $\mathcal{L}_2$ simultaneously. Another important observation is that the MAE error is more consistent than MSE. This implies the smaller errors matched up more consistently between training and validation set. Overall, we can conclude there is no significant over-fitting.
\par
In the testing phase, we generated $1000$ highly accurate samples with MC engine specification ($N$=100k, $M$=100). We test our trained network and came to the following testing results:
\begin{figure}[H]
\centering
\begin{subfigure}[b]{0.475\textwidth}
\includegraphics[width=5cm]{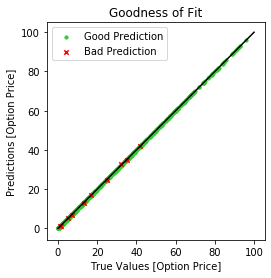}
\end{subfigure}\quad
\begin{subfigure}[b]{0.475\textwidth}
\includegraphics[width=6.2cm]{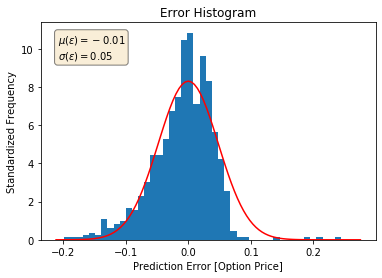}
\end{subfigure}
\end{figure}
\par
Moving on to analyzing the testing set, we observe a strong linear relationship between the predicted value and true value. This is an indication our net performs extremely well in predicting option prices. In the graph above, we observe relatively few misclassification points (error that are more than $3$ standard deviation away from the mean). Furthermore, we observes approximate normality in the residual histogram. The slight leptokurtic shape could hint hyper-parameter tuning might yield better results. However due to the close resemblances to normality, the source of error should be relatively homogeneous.
\par
We will use option parameters $s_1=60$, $s_2=80$, $\sigma_1=0.4$, $\sigma_2=0.2$, and $r=0.05$ to illustrate the capability of our trained neural net pricer in the table below:
\begin{table}[H]
\caption{FLMM Exchange Option Prediction Results}
\centering
\begin{tabular}{cccccc} 
\hline
N=1m, M=100 & $\tau=0.5$ & $\tau=1$ & $\tau=2$& Computation Time \\
\hline
$\rho=0.1$ & $1.9200259$ & $4.4097806$ & $8.4608331$& 1028.88s\\
Neural Net & $1.9289119$ & $4.3921754$ & $8.439103$ & 0.02s\\
\hline
$\rho=0.3$ & $1.4587309$ & $3.5837104$  & $7.1665319$& 992.23s\\
Neural Net & $1.4413449$ & $3.5714308$ & $7.149992$ & 0.02s\\
\hline
$\rho=0.5$ & $1.0013951$ & $2.7138871$  & $5.7575391$& 995.11\\
Neural Net & $1.00379605$ & $2.7073013$ & $5.7323675$ & 0.02s\\
\hline
$\rho=0.7$ & $0.5674676$ & $1.8049308$  & $4.2072835$& 1032.56\\
Neural Net & $0.5654936$ & $1.7964092$ & $4.183954$ & 0.02s\\
\hline
$\rho=0.9$ & $0.20259104$ & $0.88884414$ & $2.4888684$& 1012.16\\
Neural Net & $0.20034486$ & $0.8735817$ & $2.474048$ & 0.02s\\
\hline
\end{tabular}
\label{tab:Deep1}
\end{table}

\section{Conclusions}
\label{sec:conclusions}
\par
In this paper, we explored the effects of liquidity on pricing Exchange Options in a binary-asset market which we refer to as FLMM. In this market, trading only affected the price of one (the iliquid) asset. Subsequently, we established the existence and uniqueness of a strong solution for
the SDEs driving the asset prices within FLMM. By the standard replication argument we obtained a two-dimensional BS-like PDE, which
characterized the options prices. We simulated asset prices by Milstein algorithm and developed a fast-converging MC estimator with the Margrabe option as its control variate. Finally, we deployed deep learning and further improved the pricing speed.
\par
Conforming to our hypothesis, we observed the same transaction cost ``super-replication" effect as described by Liu and Yong \cite{LiunYong}. This paper may serve as a cautionary note for FX traders who regularly deal with option on iliquid currencies. Option issuers may also adopt this model as a LVA model for any type of Exchange Options.

\section{Appendix}
\label{appendix}
\par
This section will include some of the formulas and proofs left out from the main body.

\subsection{Finite Liquidity Existence and Uniqueness Theorem I}
\label{FLEUT1}
In this section, $\|\cdot\|$ and $|||\cdot|||$ represents the supremum norms:
\begin{gather*}
\|f\|=\sup_{(s_1,s_2)\in\mathcal{D}_1}|f(t,s_1,s_2)|,\quad\text{ where }\mathcal{D}_1=\big(\mathbbm{R}^+\big)^2,
\\
|||f|||=\sup_{(t,s_1,s_2)\in\mathcal{D}_2}|f(t,s_1,s_2)|,\quad\text{ where }\mathcal{D}_2=[0,T]\times\big(\mathbbm{R}^+\big)^2,
\end{gather*}
The following combination of conditions $(1)-(6)$ will guarantee existence and uniqueness of a strong solution for $S_1$.

    \begin{align*}
    (1)\qquad\|\lambda(s_1f_{s_1s_1}+s_1f_{s_1s_2}+f_{s_2}+s_2f_{s_2}+s_2f_{s_1s_1}+s_2f_{s_1s_2}+s_2f_{s_2s_2})\|<\infty.
    \end{align*}
    \begin{align*}
    (2)\qquad\|\big(\lambda_{s_1}+\lambda_{s_2}\big)\big(s_1f_{s_1}+s_2f_{s_1}+s_2f_{s_2}\big)\|<\infty.
    \end{align*}
    \begin{gather*}
    (3)\qquad|||1-\lambda{}f_{s_1}|||>\delta_0, \text{ for some }\delta_0>0. 
    \end{gather*}
    \begin{gather*}
    (4)\qquad\|(\lambda+\lambda_{s_1}+\lambda_{s_2})(f_{t}+f_{s_1}+f_{s_2}+f_{s_1s_1}+f_{s_1s_2}+f_{s_2s_2}+f_{s_1s_1s_2}
    \\
    +f_{s_1s_2s_2})\|<\infty.\nonumber
    \end{gather*}
    \begin{gather*}
    (5)\qquad\|\lambda(f_{ts_1}+f_{ts_2}+f_{s_1s_1s_1}+f_{s_2s_2s_2})+\lambda_{s_1}f_{s_1s_1s_1}+\lambda_{s_2}f_{s_2s_2s_2}\|<\infty.
    \end{gather*}
    \begin{gather*}
    (6)\qquad\|s_1f_{s_1s_1}+s_1f_{s_1s_2}+s_2f_{s_1s_2}+s_2f_{s_2s_2}+s_1^2f_{s_1s_1s_2}+s_1^2f_{s_1s_2s_2}+s_1s_2f_{s_1s_1s_2}
    \\
    +s_1s_2f_{s_1s_2s_2}+s_2^2f_{s_1s_2s_2}+s_2^2f_{s_2s_2s_2}\|<\infty.\nonumber
    \end{gather*} 

%In this section, $\|\cdot\|$ represents the Supremum norm:
%\begin{gather*}
%\|f\|=\sup_{(t,s_1,s_2)\in\mathcal{D}}|f(t,s_1,s_2)|,\quad\text{ where }\mathcal{D}=[0,T]\times\big(\mathbbm{R}^+\big)^2.
%\end{gather*}
\begin{proof}
\label{sec:FiniteExist}
Recall that the SDE of $S_1$ is of the form:
    \begin{align*}
    dS_1(t)&=\bar{\mu}_1\big(\mathbf{S}(t)\big)dt+\bar{\sigma}_{11}\big(\mathbf{S}(t)\big)dW_1(t)+\bar{\sigma}_{12}\big(\mathbf{S}(t)\big)dW_2(t),
    \end{align*}
    where
    \begin{align*}
    &\bar{\mu}_1(t,s_1,s_2)=\frac{1}{1-\lambda{}f_{s_1}}\Big(\mu_1s_1+\lambda{}f_{t}+s_2\mu_2\lambda{}f_{s_2}+\frac{f_{s_1s_2}(\rho\sigma_1\sigma_2s_1s_2+\sigma_2^2s_2^2\lambda{}f_{s_2})}{1-\lambda{}f_{s_1}}
    \\
    &\qquad\qquad\quad+\frac{f_{s_1s_1}(\sigma_1^2s_1^2+\sigma^2_2s_2^2\lambda^2f_{s_2}^2+2\rho\sigma_1\sigma_2s_1s_2\lambda{}f_{s_2})}{2\big(1-\lambda{}f_{s_1}\big)^2}+\frac{\sigma_2^2s_2^2f_{s_2s_2}}{2}\Big),
    \\
    &\bar{\sigma}_{11}(t,s_1,s_2)=\frac{\sigma_1s_1}{1-\lambda{}f_{s_1}},\qquad\bar{\sigma}_{12}(t,s_1,s_2)=\frac{\sigma_2s_2\lambda{}f_{s_2}}{1-\lambda{}f_{s_1}}.
    \end{align*}
\par
Following the classical existence uniqueness result for SDEs, we have to show the functions $\bar{\mu}_{1}(t,s_1,s_2)$, $\bar{\sigma}_{11}(t,s_1,s_2)$ and $\bar{\sigma}_{12}(t,s_1,s_2)$ are uniformly Lipschitz continuous with respect to $\|\cdot\|$. Thus, it is sufficient to check the boundedness of their respective partial derivatives. Computing the derivatives, we have:
\begin{align*}
\big[\bar{\sigma}_{11}\big]_{s_1}&=\sigma_1\big(\frac{1}{1-\lambda{}f_{s_1}}+\frac{s_1(\lambda_{s_1}f_{s_1}+\lambda{}f_{s_1s_1})}{(1-\lambda{}f_{s_1})^2}\big),
\\
\big[\bar{\sigma}_{11}\big]_{s_2}&=\sigma_1s_1\frac{\lambda_{s_2}f_{s_1}+\lambda{}f_{s_1s_2}}{(1-\lambda{}f_{s_1})^2},
\\
\big[\bar{\sigma}_{12}\big]_{s_1}&=\sigma_2s_2\big(\frac{(\lambda_{s_1}f_{s_2}+\lambda{}f_{s_1s_2})}{1-\lambda{}f_{s_1}}+\frac{\lambda{}f_{s_2}(\lambda{}f_{s_1s_1}+\lambda_{s_1}f_{s_1})}{(1-\lambda{}f_{s_1})^2}\big),
\\
\big[\bar{\sigma}_{12}\big]_{s_2}&=\sigma_2\big(\frac{\lambda{}f_{s_2}+s_2(\lambda_{s_2}f_{s_2}+\lambda{}f_{s_2s_2})}{1-\lambda{}f_{s_1}}+\lambda\frac{s_2f_{s_2}(\lambda_{s_2}f_{s_1}+\lambda{}f_{s_1s_2})}{(1-\lambda{}f_{s_1})^2}\big).
\end{align*}
We can clearly see the boundedness requirement for $\big[\bar{\sigma}_{11}\big]_{s_1}$, $\big[\bar{\sigma}_{12}\big]_{s_1}$, $\big[\bar{\sigma}_{11}\big]_{s_2}$ and $\big[\bar{\sigma}_{12}\big]_{s_2}$ can be condensed into:
\begin{gather}\label{itocond1}
\|\lambda(s_1f_{s_1s_1}+s_1f_{s_1s_2}+f_{s_2}+s_2f_{s_2}+s_2f_{s_1s_1}+s_2f_{s_1s_2}+s_2f_{s_2s_2})\|<\infty,
\\
\label{itocond2}
\|\big(\lambda_{s_1}+\lambda_{s_2}\big)\big(s_1f_{s_1}+s_2f_{s_1}+s_2f_{s_2}\big)\|<\infty.
\end{gather}
Furthermore, we will require the denominator terms in the partial derivatives above to satisfy:
\begin{gather}\label{itocond3}
|||1-\lambda{}f_{s_1}|||>\delta_0, \text{ for some }\delta_0>0. 
\end{gather}
The partial derivatives $\big[\bar{\mu}_1s_1\big]_{s_1}$ and $\big[\bar{\mu}_1s_1\big]_{s_2}$ are:
\begin{align*}
\big[\bar{\mu}_1s_1\big]_{s_1}&=\mu_1\big(\frac{1}{1-\lambda{}f_{s_1}}+\frac{s_1(\lambda_{s_1}f_{s_1}+\lambda{}f_{s_1s_1})}{(1-\lambda{}f_{s_1})^2}\big)+\frac{\lambda{}f_{ts_1}+\lambda_{s_1}f_{t}}{1-\lambda{}f_{s_1}}
\\
&+\frac{\lambda{}f_{t}(\lambda_{s_1}f_{s_1}+\lambda{}f_{s_1s_1})}{(1-\lambda{}f_{s_1})^2}+\mu_2s_2\big(\frac{\lambda{}f_{s_1s_2}+\lambda_{s_1}f_{s_2}}{1-\lambda{}f_{s_1}}+\frac{\lambda{}f_{s_2}(\lambda_{s_1}f_{s_1}+\lambda{}f_{s_1s_1})}{(1-\lambda{}f_{s_1})^2}\big)
\\
&+\frac{1}{2}\sigma_1^2\big(\frac{s_1^2f_{s_1s_1s_1}+2s_1f_{s_1s_1}}{(1-\lambda{}f_{s_1})^3}+3\frac{s_1^2f_{s_1s_1}(\lambda_{s_1}f_{s_1}+\lambda{}f_{s_1s_1})}{(1-\lambda{}f_{s_1})^4}\big)
\\
&+\frac{1}{2}\sigma_2^2s_2^2\big(\frac{\lambda^2f_{s_2}^2f_{s_1s_1s_1}+2\lambda{}f_{s_2}f_{s_1s_1}(\lambda_{s_1}f_{s_2}+\lambda{}f_{s_1s_2})}{(1-\lambda{}f_{s_1})^3}
\\
&+3\frac{\lambda^2f_{s_2}^2f_{s_1s_1}(\lambda_{s_1}f_{s_1}+\lambda{}f_{s_1s_1})}{(1-\lambda{}f_{s_1})^4}\big)
\\
&+\rho\sigma_1\sigma_2s_2\big(\frac{s_1\lambda{}f_{s_2}f_{s_1s_1s_1}+s_1f_{s_1s_1}(\lambda_{s_1}f_{s_2}+\lambda{}f_{s_1s_2})+\lambda{}f_{s_1s_1}f_{s_2}}{(1-\lambda{}f_{s_1})^3}
\\
&+3\frac{s_1\lambda{}f_{s_2}f_{s_1s_1}(\lambda_{s_1}f_{s_1}+\lambda{}f_{s_1s_1})}{(1-\lambda{}f_{s_1})^4}\big)
\\
&+\rho\sigma_1\sigma_2s_2\big(\frac{f_{s_1s_2}+s_1f_{s_1s_1s_2}}{(1-\lambda{}f_{s_1})^2}+2\frac{s_1(\lambda_{s_1}f_{s_1}+\lambda{}f_{s_1s_1})}{(1-\lambda{}f_{s_1})^3}\big)
\\
&+\sigma_2^2s_2^2\big(\frac{\lambda{}f_{s_2}f_{s_1s_1s_2}+f_{s_1s_2}(\lambda_{s_1}f_{s_1}+\lambda{}f_{s_1s_1})}{(1-\lambda{}f_{s_1})^2}+\frac{\lambda{}f_{s_2}f_{s_1s_2}(\lambda_{s_1}f_{s_1}+\lambda{}f_{s_1s_1})}{(1-\lambda{}f_{s_1})^3}\big)
\\
&+\frac{1}{2}\sigma_2^2s_2^2\big(\frac{f_{s_1s_2s_2}}{1-\lambda{}f_{s_1}}+\frac{f_{s_2s_2}(\lambda_{s_1}f_{s_1}+\lambda{}f_{s_1s_1})}{(1-\lambda{}f_{s_1})^2}\big).
\\
\big[\bar{\mu}_1s_1\big]_{s_2}=&\mu_1s_1\frac{\lambda_{s_2}f_{s_1}+\lambda{}f_{s_1s_2}}{(1-\lambda{}f_{s_1})^2}+\frac{\lambda_{s_2}f_{t}+\lambda{}f_{ts_2}}{1-\lambda{}f_{s_1}}+\frac{\lambda{}f_t(\lambda_{s_2}f_{s_1}+\lambda{}f_{s_1s_2})}{(1-\lambda{}f_{s_1})^2}
\\
+&\mu_2\big(\frac{\lambda{}f_{s_2}+s_2(\lambda_{s_2}f_{s_2}+\lambda{}f_{s_1s_2})}{1-\lambda{}f_{s_1}}+\frac{s_2\lambda{}f_{s_2}(\lambda_{s_2}f_{s_1}+\lambda{}f_{s_1s_2})}{(1-\lambda{}f_{s_1})^2}\big)
\\
+&\frac{1}{2}\sigma_1^2s_1^2\big(\frac{f_{s_1s_1s_2}}{(1-\lambda{}f_{s_1})^3}+3\frac{f_{s_1s_1}(\lambda_{s_2}f_{s_1}+\lambda{}f_{s_1s_2})}{(1-\lambda{}f_{s_1})^4}\big)
\\
+&\frac{1}{2}\sigma_2^2\Big(\frac{\lambda^2f_{s_1s_1}f_{s_2}^2+s_2\big(\lambda^2+f_{s_1s_1s_2}f_{s_2}^2+f_{s_1s_1}(2\lambda\lambda_{s_2}f_{s_2}^2+s\lambda^2f_{s_2}f_{s_2s_2})\big)}{(1-\lambda{}f_{s_1})^3}
\\
+&3\frac{s_2\lambda^2f_{s_1s_1}f_{s_2}^2(\lambda_{s_2}f_{s_1}+\lambda{}f_{s_1s_2})}{(1-\lambda{}f_{s_1})^4}\Big)
\\
+&\rho\sigma_1\sigma_2s_1\Big(\frac{\lambda{}f_{s_2}f_{s_1}^2+s_2\big(\lambda_{s_2}f_{s_2}f_{s_1}^2+\lambda(f_{s_2s_2}f_{s_1}^2+2f_{s_2}f_{s_1}f_{s_1s_2})\big)}{(1-\lambda{}f_{s_1})^3}
\\
+&3\frac{s_2\lambda{}f_{s_2}f_{s_1s_1}(\lambda_{s_2}f_{s_1}+\lambda{}f_{s_1s_2})}{(1-\lambda{}f_{s_1})^4}\Big)
\\
+&\rho\sigma_1\sigma_2s_1\big(\frac{f_{s_1s_2}+s_2f_{s_1s_2s_2}}{(1-\lambda{}f_{s_1})^2}+\frac{s_2f_{s_1s_2}(\lambda_{s_2}f_{s_1}+\lambda{}f_{s_1s_2}}{(1-\lambda{}f_{s_1})^3}\big)
\\
+&\sigma_2^2\Big(\frac{2s_2^2\lambda{}f_{s_2}f_{s_1s_2}+s_2^2\big(\lambda_{s_2}f_{s_2}f_{s_1s_2}+\lambda(f_{s_1s_2}f_{s_2s_2}+f_{s_2}f_{s_1s_2s_2})\big)}{(1-\lambda{}f_{s_1})^2}
\\
+&\frac{s_2^2\lambda{}f_{s_2}f_{s_1s_2}(\lambda_{s_2}f_{s_1}+\lambda{}f_{s_1s_2})}{(1-\lambda{}f_{s_1})^3}\Big)
\\
+&\frac{1}{2}\sigma_2^2\big(\frac{2s_2f_{s_2s_2}+s_2^2f_{s_2s_2s_2}}{1-\lambda{}f_{s_1}}+\frac{s_2^2f_{s_2s_2}(\lambda_{s_2}f_{s_1}+\lambda{}f_{s_1s_2})}{(1-\lambda{}f_{s_1})^2}\big).
\end{align*}
We conclude the partial derivatives of $\bar{\mu}_1(t,s_1,s_2)$ will be bonded when $|||1-\lambda{}f_{s_1}|||>\delta_0$ and:
\begin{gather}
\label{itocond4}
\|(\lambda+\lambda_{s_1}+\lambda_{s_2})(f_{t}+f_{s_1}+f_{s_2}+f_{s_1s_1}+f_{s_1s_2}+f_{s_2s_2}+f_{s_1s_1s_2}
\\
+f_{s_1s_2s_2})\|<\infty,\nonumber
\\
\label{itocond5}
\|\lambda(f_{ts_1}+f_{ts_2}+f_{s_1s_1s_1}+f_{s_2s_2s_2})+\lambda_{s_1}f_{s_1s_1s_1}+\lambda_{s_2}f_{s_2s_2s_2}\|<\infty,
\\
\label{itocond6}
\|s_1f_{s_1s_1}+s_1f_{s_1s_2}+s_2f_{s_1s_2}+s_2f_{s_2s_2}+s_1^2f_{s_1s_1s_2}+s_1^2f_{s_1s_2s_2}+s_1s_2f_{s_1s_1s_2}
\\
+s_1s_2f_{s_1s_2s_2}+s_2^2f_{s_1s_2s_2}+s_2^2f_{s_2s_2s_2}\|<\infty,\nonumber
\end{gather}
\par
The combination of requirements \eqref{itocond1}, \eqref{itocond2}, \eqref{itocond3}, \eqref{itocond4}, \eqref{itocond5}, \eqref{itocond6} will guarantee $s_1\bar{\mu}_{1}(t,s_1,s_2)$, $s_1\bar{\sigma}_{11}(t,s_1,s_2)$ and $s_1\bar{\sigma}_{12}(t,s_1,s_2)$ are uniformly Lipschitz continuous in $\big(\mathbb{R}^+\big)^2$. By It\^o's Existence and Uniqueness Theorem It\^o \cite{Ito} (1979), the SDE for $S_1$ will have a unique strong solution.
\end{proof}

\subsection{Finite Liquidity Existence and Uniqueness Theorem II}
\label{sec:FLEUT2}
To show the SDE have a unique strong solution, it is sufficient to show that the conditions $(1)-(6)$ in Appendix Section \ref{FLEUT1} are satisfied for the particular choice of $\lambda(t,s_1)$ and $f(t,s_1,s_2)=\Delta_1(t)$.
    \begin{itemize}
    \item Condition (1):
            \begin{align*}
            &\|\lambda(s_1Spd_{111}+s_1Spd_{112}+\Gamma_{12}+s_2\Gamma_{12}+s_2Spd_{111}+s_2Spd_{112}+s_2Spd_{122})\|
            \\
            &=\|\lambda\Big(\frac{N'(d_+)}{\sigma{}s_1\sqrt{\tau}}\big(\frac{2d_+}{\sigma{}s_1\sqrt{\tau}}+1\big)+s_1\frac{2d_+N'(d_+)}{\sigma^2\tau{}s_1s_2}+\frac{1}{\sigma\sqrt{\tau}}\frac{s_2N'(d_+)}{s_1^2}+\frac{s_2^2N'(d_+)}{\sigma\sqrt{\tau}s_1^2}
            \\
            &+\frac{s_2N'(d_+)}{\sigma{}s_1^2\sqrt{\tau}}\big(\frac{2d_+}{\sigma{}s_1\sqrt{\tau}}+1\big)+\frac{2d_+N'(d_+)}{\sigma^2\tau{}s_1}+\frac{2d_-s_2N'(d_+)}{\sigma^2\tau{}s_1^2}\Big)\|<\infty.
            \end{align*}
            \begin{proof}
            Notice there is a common term of the form $\frac{N'(d_+)}{s_1^n}$. These terms appears naturally in higher order Greeks. Consider any real number $n$, we have:
            \begin{align*}
            \frac{N'(d_+)}{s_1^n}&=\frac{1}{s_1^n\sqrt{2\pi}}\exp{\Big\{-\Big(\frac{\log(\frac{s_1}{s_2})+\frac{1}{2}\sigma^2\tau}{\sigma\sqrt{\tau}}\Big)^2\Big\}}
            \\
            &=\frac{1}{s_1^n\sqrt{2\pi}}e^{\Big\{-\frac{\log^2(s_1)+\log(s_1)\big(\frac{1}{2}\sigma^2\tau-\log(s_2)\big)+\big(\frac{1}{2}\sigma^2\tau-\log(s_2)\big)^2}{\sigma^2\tau}\Big\}}e^{-n\log(s_1)}
            \\
            &=\frac{1}{\sqrt{2\pi}}\exp{\Big\{-\frac{\log^2(s_1)+o\big(\log(s_1)\big)}{\sigma^2\tau}\Big\}},
            \end{align*}
            which approaches to $0$ as $s_1$ approaches to zero, and approaches to $0$ as well as $s_1$ approaches $\infty$. Since $n$ was arbitrary, then all of the functions in Condition $(1)$ are bounded in $s_1$. With a similar method involving the common term $\frac{N'(d_+)}{s_2^n}$, we can also show that all of the terms in Condition $(1)$ are bounded in $s_2$. We can ultimately conclude that the entire function of Condition $(1)$ is bounded in $(s_1,s_2)$.
            \end{proof}
    \item Condition (2):
            \begin{align*}
            \|\lambda_{s_1}\big(s_1\Gamma_{11}+s_2\Gamma_{11}+s_2\Gamma_{12}\big)\|<\infty.
            \end{align*}
            \begin{proof}
            Same proof as Condition (1).
            \end{proof}
    \item Condition (3):
            \begin{gather*}
            |||1-\lambda\Gamma_{11}|||>\delta_0, \text{ for some }\delta_0>0. 
            \end{gather*}
            \begin{proof}
            This condition already holds in the $s_1, s_2$ dimension. For $t$ we have $\lim_{t \to T}\bar{\lambda}(t,s_1)=0$ and $\lim_{t \to T} \Gamma_{11}(t)=\infty$ for at the money options. Since $\bar{\lambda}(t,s_1)$ approach to $0$ at a greater rate, then $\lim_{t\to T}\bar{\lambda}(t,s_1)\Gamma_{11}(t)=0.$ In fact, this ensures the $\bar{\lambda}(t,s_1)\Gamma_{11}(t)$ term stays small, which ultimately guarantees the existence of $\delta_0$. There is a more detailed explanation in Pirvu et al (2014) \cite{Pirvu}.
            \end{proof}
    \item Condition (4):
            \begin{align*}
            &\|(\lambda+\lambda_{s_1})(Chm_1+\Gamma_{11}+\Gamma_{12}+Spd_{111}+Spd_{112}+Spd_{122}+Acc_{1112}
            \\
            &+Acc_{1122})\|<\infty.\nonumber 
            \end{align*}
            \begin{proof}
            Same proof as Condition (1).
            \end{proof}
    \item Condition (5):
            \begin{gather*}
            \|\lambda(Col_1+Col_2+Acc_{1111}+Acc_{1222})+\lambda_{s_1}Acc_{1111}+\lambda_{s_2}Acc_{1222}\|<\infty.
            \end{gather*}
            \begin{proof}
            Same proof as Condition (1).
            \end{proof}
    \item Condition (6):
            \begin{align*}
            &\|s_1Spd_{111}+s_1Spd_{112}+s_2Spd_{112}+s_2Spd_{122}+s_1^2Acc_{1112}+s_1^2Acc_{1122}
            \\
            &+s_1s_2Acc_{1112}+s_1s_2Acc_{1122}+s_2^2Acc_{1122}+s_2^2Acc_{1222}\|<\infty. 
            \end{align*}
            \begin{proof}
            Same proof as Condition (1).
            \end{proof}
    \end{itemize}
Since we have shown Condition $(1)$ to $(5)$ in the Appendix Section \ref{FLEUT1} holds for our price impact trading strategy $\lambda\big(t,S_1(t)\big)df\big(t,S_1(t),S_2(t)\big)$. We can conclude the SDEs $S_1$ \eqref{FLMMSDEexchange} has a strong solution.

\subsection{Margrabe's Pricing Formula and Greeks}
\par
Margrabe (1978) \cite{Margrabe} derived the following closed form price for Exchange Option.
\begin{gather}
\label{Margrabe}
V(t,s_1,s_2)=\widetilde{\mathbb{E}}\big[e^{-r\tau}\big(S_1(T)-S_2(T)\big)^+|\mathscr{F}(t)\big]=s_1N(d_+)-s_2N(d_-),
\\
\text{where }d_\pm=\frac{\log(\frac{s_1}{s_2})\pm\frac{1}{2}\sigma^2\tau}{\sigma\sqrt{\tau}},\text{ and }\sigma^2=\sigma_1^2+\sigma_2^2-2\sigma_1\sigma_2\rho.\nonumber
\end{gather}
We can derive the Exchange Option Greeks by differentiating formula \eqref{Margrabe}. The first order Greeks are well known, they are available in papers such as Alos and Thorsten (2017) \cite{Alos1}.
\begin{gather*}
\Delta_1(t)=N(d_+)\qquad\Delta_2(t)=-N(d_-).
\\
\Theta(t)=\frac{\sigma{}s_1N'(d_+)}{2\sqrt{\tau}}=-\frac{\sigma{}s_2N'(d_-).}{2\sqrt{\tau}}.
\end{gather*}
For the second and higher order Greeks, we will provide derivations.
\begin{align*}
\label{sec:greek}
&\Gamma_{11}(t)=\frac{\partial{\Delta_{1}(t)}}{\partial{s_1}}=N'(d_+)\frac{\partial{d_+}}{\partial{s_1}}=\frac{N'(d_+)}{\sigma s_1\sqrt{\tau}}.
\\
&\Gamma_{22}(t)=\frac{\partial{\Delta_{2}(t)}}{\partial{s_2}}=-N'(d_-)\frac{\partial{d_-}}{\partial{s_2}}=\frac{N'(d_-)}{\sigma s_2\sqrt{\tau}}.
\\
&\Gamma_{12}(t)=\Gamma_{21}(t)=\frac{\partial{\Delta_{1}(t)}}{\partial{s_2}}=N'(d_+)\frac{\partial{d_+}}{\partial{s_2}}=-\frac{1}{\sigma\sqrt{\tau}}\frac{N'(d_+)}{s_2}=-\frac{1}{\sigma\sqrt{\tau}}\frac{N'(d_-+\sigma\sqrt{\tau})}{s_2}
\\
&\quad=-\frac{1}{\sigma\sqrt{\tau}}\frac{1}{s_2}\frac{1}{\sqrt{2\pi}}\exp\Big\{-\frac{1}{2}d_-^2-d_-\sigma\sqrt{\tau}-\frac{1}{2}\sigma^2\tau\Big\}
\\
&\quad=-\frac{1}{\sigma\sqrt{\tau}}\frac{1}{s_2}\frac{1}{\sqrt{2\pi}}\exp\Big\{-\frac{1}{2}d_-^2-\log{\big(\frac{s_1}{s_2}\big)}\Big\}=-\frac{N'(d_-)}{\sigma s_1\sqrt{\tau}}.
\\
&Charm_1(t)=\frac{\partial{\Delta_1(t)}}{\partial\tau}=N'(d_+)\frac{\partial{d_+}}{\partial\tau}=N'(d_+)\Big(-\frac{\log\big(\frac{s_1}{s_2}\big)}{2\sigma\tau^\frac{3}{2}}+\frac{\sigma}{4\sqrt{\tau}}\Big).
\\
&Charm_2(t)=\frac{\partial{\Delta_2(t)}}{\partial\tau}=-N'(d_-)\frac{\partial{d_-}}{\partial\tau}=N'(d_-)\Big(\frac{\log\big(\frac{s_1}{s_2}\big)}{2\sigma\tau^\frac{3}{2}}+\frac{\sigma}{4\sqrt{\tau}}\Big).
\\
&Speed_{111}(t)=\frac{\partial{\Gamma_{11}(t)}}{\partial{s_1}}=\frac{1}{\sigma\sqrt{\tau}}\frac{N''(d_+)\frac{\partial{d_+}}{\partial{s_1}}-N'(d_+)}{s_1^2}=\frac{1}{\sigma\sqrt{\tau}}\frac{-\frac{2d_+N'(d_+)}{\sigma{}s_1\sqrt{\tau}}-N'(d_+)}{s_1^2}
\\
&\quad=-\frac{\Gamma_{11}}{s_1}\big(\frac{2d_+}{\sigma{}s_1\sqrt{\tau}}+1\big).
\\
&Speed_{222}(t)=\frac{\partial{\Gamma_{22}(t)}}{\partial{s_2}}=\frac{1}{\sigma\sqrt{\tau}}\frac{N''(d_-)\frac{\partial{d_-}}{\partial{s_2}}-N'(d_-)}{s_2^2}=\frac{1}{\sigma\sqrt{\tau}}\frac{-\frac{2d_-N'(d_-)}{\sigma{}s_2\sqrt{\tau}}-N'(d_-)}{s_2^2}
\\
&\quad=-\frac{\Gamma_{22}}{s_2}\big(\frac{2d_-}{\sigma{}s_2\sqrt{\tau}}+1\big).
\\
&Speed_{112}(t)=\frac{\partial{\Gamma_{11}(t)}}{\partial{s_2}}=\frac{1}{\sigma\sqrt{\tau}}\frac{N''(d_+)\frac{\partial{d_+}}{\partial{s_2}}}{s_1}=-\frac{2d_+N'(d_+)}{\sigma^2\tau s_1s_2}=-\frac{2d_+\Gamma_{11}}{\sigma s_2}.
\\
&Speed_{221}(t)=\frac{\partial{\Gamma_{22}(t)}}{\partial{s_1}}=\frac{1}{\sigma\sqrt{\tau}}\frac{N''(d_-)\frac{\partial{d_-}}{\partial{s_1}}}{s_2}=-\frac{2d_-N'(d_-)}{\sigma^2\tau s_1s_2}=-\frac{2d_-\Gamma_{22}}{\sigma s_1}.
\\
&Colour_{11}(t)=\frac{\partial{\Gamma_{11}(t)}}{\partial\tau}=\frac{1}{\sigma{}s_1}\Big(-\frac{1}{2\tau^{\frac{3}{2}}}N'(d_+)-\frac{1}{\tau^{\frac{1}{2}}}N'(d_+)d_+\frac{\partial{d_+}}{\partial\tau}\Big)
\\
&\quad=\frac{N'(d_+)}{2\sigma\tau^\frac{3}{2}s_1}\Big\{-1+d_+\Big(\log(\frac{s_1}{s_2})\frac{1}{\sigma\sqrt{\tau}}-\frac{1}{2}\sigma\sqrt{\tau}\Big)\Big\}
\\
&\quad=-\frac{\Gamma_{11}}{2^3\sigma^2\tau^{2}}\Big(\sigma^4\tau^2+4\sigma^2\tau-4\log^2(\frac{s_1}{s_2})\Big),
\\
&Colour_{22}(t)=\frac{\partial{\Gamma_{22}(t)}}{\partial\tau}=\frac{1}{\sigma{}s_2}\Big(-\frac{1}{2\tau^{\frac{3}{2}}}N'(d_-)-\frac{1}{\tau^{\frac{1}{2}}}N'(d_-)d_+\frac{\partial{d_-}}{\partial\tau}\Big)
\\
&\quad=-\frac{\Gamma_{22}}{2^3\sigma^2\tau^{2}}\Big(\sigma^4\tau^2+4\sigma^2\tau-4\log^2(\frac{s_1}{s_2})\Big),
\\
&Colour_{12}(t)=Colour_{21}(t)=\frac{\partial{\Gamma_{12}(t)}}{\partial\tau}=\frac{1}{\sigma{}s_2}\Big(\frac{1}{2\tau^{\frac{3}{2}}}N'(d_+)+\frac{1}{\tau^{\frac{1}{2}}}N'(d_+)d_+\frac{\partial{d_+}}{\partial\tau}\Big)
\\
&\quad=\frac{-\Gamma_{12}}{2^3\sigma^2\tau^{2}}\Big(\sigma^4\tau^2+4\sigma^2\tau-4\log^2(\frac{s_1}{s_2})\Big)=\frac{-\Gamma_{21}}{2^3\sigma^2\tau^{2}}\Big(\sigma^4\tau^2+4\sigma^2\tau-4\log^2(\frac{s_1}{s_2})\Big).
\end{align*}
\begin{align*}
&Acceleration_{1111}(t)
\\
&\quad=\frac{\partial{Speed_{111}(t)}}{\partial{s_1}}=-\Big(\big(\frac{\partial}{\partial{s_1}}{\frac{\Gamma_{11}}{s_1}}\big)\big(\frac{2d_+}{\sigma{}s_1\sqrt{\tau}}+1\big)+\frac{\Gamma_{11}}{s_1}\frac{2}{\sigma\sqrt{\tau}}\big(\frac{\partial}{\partial{s_1}}\frac{d_+}{s_1}\big)\Big)
\\
&\quad=-\Big(\frac{Speed_{111}s_1-\Gamma_{11}}{s_1^2}\big(\frac{2d_+}{\sigma{}s_1\sqrt{\tau}}+1\big)+\frac{\Gamma_{11}}{s_1}\frac{2}{\sigma\sqrt{\tau}}\big(\frac{\frac{1}{\sigma\sqrt{\tau}}-d_+}{s_1^2}\big)\Big)
\\
&\quad=-\frac{2\Gamma_{11}}{\sigma\sqrt{\tau}s_1^3}\Big(d_+\big(\frac{2d_+}{\sigma{}\sqrt{\tau}s_1}+1\big)+\big(\frac{1}{\sigma\sqrt{\tau}}-d_+\big)\Big)
\\
&\quad=-\frac{2\Gamma_{11}}{\sigma^2s_1^3\tau}\big(\frac{2d_+^2}{s_1}+1\big),
\\
&Acceleration_{1112}(t)
\\
&\quad=\frac{\partial{Speed_{111}(t)}}{\partial{s_2}}=-\Big(\frac{1}{s_1}\big(\frac{\partial}{\partial{s_2}}\Gamma_{11}\big)\big(\frac{2d_+}{\sigma{}s_1\sqrt{\tau}}+1\big)+\frac{\Gamma_{11}}{s_1}\frac{2}{\sigma\sqrt{\tau}s_1}\big(\frac{\partial}{\partial{s_2}}d_+\big)\Big)
\\
&\quad=-\Big(\frac{Speed_{112}}{s_1}\big(\frac{2d_+}{\sigma{}\sqrt{\tau}s_1}+1\big)-\frac{\Gamma_{11}}{s_1}\frac{2}{\sigma\sqrt{\tau}s_1}\frac{1}{\sigma\sqrt{\tau}s_2}\Big)
\\
&\quad=\frac{2\Gamma_{11}}{\sigma s_1s_2\sqrt{\tau}}\big(\frac{2d_+^2}{\sigma \sqrt{\tau}}+d_++\frac{1}{\sigma s_1\sqrt{\tau}}\big),
\\
&Acceleration_{1122}(t)
\\
&\quad=\frac{\partial{Speed_{112}(t)}}{\partial{s_2}}=\frac{2}{\sigma\sqrt{\tau}s_1}\big(\frac{\partial}{\partial{s_2}}d_+\big)\Gamma_{12}+d_+\big(\frac{\partial}{\partial{s_2}}\Gamma_{12}\big)
\\
&\quad=\frac{2\Gamma_{12}}{\sigma^2s_1s_2\tau}\big(d_+d_--1\big),
\\
&Acceleration_{1222}(t)
\\
&\quad=\frac{\partial{Speed_{222}(t)}}{\partial{s_1}}=-\Big(\frac{1}{s_2}\big(\frac{\partial}{\partial{s_1}}\Gamma_{22}\big)\big(\frac{2d_-}{\sigma{}s_2\sqrt{\tau}}+1\big)+\frac{\Gamma_{22}}{s_2}\frac{2}{\sigma\sqrt{\tau}s_2}\big(\frac{\partial}{\partial{s_1}}d_-\big)\Big)
\\
&\quad=-\Big(\frac{Speed_{122}}{s_2}\big(\frac{2d_-}{\sigma{}\sqrt{\tau}s_2}+1\big)+\frac{\Gamma_{22}}{s_2}\frac{2}{\sigma\sqrt{\tau}s_2}\frac{1}{\sigma\sqrt{\tau}s_1}\Big)
\\
&\quad=\frac{2\Gamma_{22}}{\sigma s_1s_2\sqrt{\tau}}\big(\frac{2d_-^2}{\sigma s_2\sqrt{\tau}}+d_--\frac{1}{\sigma s_2\sqrt{\tau}}\big),
\end{align*}
\begin{align*}
&Acceleration_{2222}(t)
\\
&\quad=\frac{\partial{Speed_{222}(t)}}{\partial{s_2}}=-\Big(\big(\frac{\partial}{\partial{s_2}}\frac{\Gamma_{22}}{s_2}\big)\big(\frac{2d_-}{\sigma{}s_2\sqrt{\tau}}+1\big)+\frac{\Gamma_{22}}{s_2}\frac{2}{\sigma\sqrt{\tau}}\big(\frac{\partial}{\partial{s_2}}\frac{d_-}{s_2}\big)\Big)
\\
&\quad=-\frac{2\Gamma_{22}}{\sigma^2 s_2^3\tau}\big(\frac{2d_-^2}{s_2}+1\big).
\end{align*}
\section*{Acknowledgments}
The authors are grateful to the anonymous referee for a careful checking of the details and for helpful comments that improved this paper.

\bibliographystyle{siamplain}
\bibliography{main}

\begin{thebibliography}{10}

\bibitem{Ahmadian}
{\sc D.~Ahmadian, O.~F. Rouz, K.~Ivaz, and A.~Safdari-Vaighani}, {\em Robust
  numerical algorithm to the european option with illiquid markets}, Applied
  Mathematics and Computation, 366 (2020), p.~124693,
  \url{http://www.sciencedirect.com/science/article/pii/S009630031930685X}.

\bibitem{Alos2}
{\sc E.~Alòs and M.~Coulon}, {\em On the optimal choice of strike conventions
  in exchange option pricing}, EconPapers,  (2018),
  \url{https://arxiv.org/abs/1807.05396}.

\bibitem{Alos1}
{\sc E.~Alòs and T.~Rheinländer}, {\em Pricing and hedging margrabe options
  with stochastic volatilities}, EconPapers,  (2017),
  \url{https://EconPapers.repec.org/RePEc:upf:upfgen:1475}.

\bibitem{Arenas}
{\sc A.~J. Arenas, G.~Gonzalez-Parra, and B.~M. Caraballo}, {\em A nonstandard
  finite difference scheme for a nonlinear black-scholes equation},
  Mathematical and Computer Modelling, 57 (2013), pp.~1663 -- 1670,
  \url{http://www.sciencedirect.com/science/article/pii/S0895717711006947}.

\bibitem{Culkin}
{\sc R.~Culkin and S.~R. Das}, {\em Machine learning in finance: The case of
  deep learning for option pricing}, Journal of Investment Management,  (2017).

\bibitem{Dyshaev}
{\sc M.~Dyshaev and V.~Fedorov}, {\em The sensitivities (greeks) for some
  models of option pricing with market illiquidity}, Mathematical notes of
  NEFU, 26 (2019), \url{https://doi.org/10.25587/SVFU.2019.102.31514}.

\bibitem{Ferguson}
{\sc R.~Ferguson and A.~D. Green}, {\em Applying deep learning to derivatives
  valuation}, SSRN Electronic Journal,  (2018),
  \url{http://dx.doi.org/10.2139/ssrn.3244821}.

\bibitem{Friedman}
{\sc A.~Friedman}, {\em Stochastic Differential Equations and Applications},
  Academic Press, 1st edition~ed., 1975.

\bibitem{Glassman}
{\sc M.~Giles and P.~Glasserman}, {\em Smoking adjoints: fast evaluation of
  greeks in monte carlo calculations}, Risk Journals,  (2005).

\bibitem{Giles2}
{\sc M.~B. Giles and L.~Szpruch}, {\em Multilevel monte carlo methods for
  applications in finance}, High-Performance Computing in Finance,  (2018),
  pp.~197--247, \url{https://arxiv.org/pdf/1212.1377.pdf}.

\bibitem{glassbook}
{\sc P.~Glasserman}, {\em Monte Carlo methods in financial engineering},
  Springer, 2004.

\bibitem{Glover}
{\sc K.~J. Glover, P.~W. Duck, and D.~P. Newton}, {\em On nonlinear models of
  markets with finite liquidity: Some cautionary notes}, SIAM Journal on
  Applied Mathematics, 70 (2010), pp.~3252--3271,
  \url{https://doi.org/10.1137/080736119}.

\bibitem{Hainaut}
{\sc D.~Hainaut}, {\em Calendar spread exchange options pricing with gaussian
  random fields}, Risks, 6 (2018), p.~77,
  \url{https://doi.org/10.3390/risks6030077}.

\bibitem{Desmond}
{\sc D.~J. Higham}, {\em An introduction to multilevel monte carlo for option
  valuation}, International Journal of Computer Mathematics, 92 (2015),
  pp.~2347--2360, \url{https://doi.org/10.1080/00207160.2015.1077236}.

\bibitem{HORNIK}
{\sc K.~Hornik, M.~Stinchcombe, and H.~White}, {\em Multilayer feedforward
  networks are universal approximators}, Neural Networks, 2 (1989), pp.~359 --
  366, \url{http://www.sciencedirect.com/science/article/pii/0893608089900208}.

\bibitem{Ito}
{\sc I.~It\^{o}}, {\em On the existence and uniqueness of solutions of
  stochastic integral equations of the volterra type}, Kodai Math, 2 (1979),
  pp.~158--170, \url{https://doi.org/https://doi.org/10.2996/kmj/1138036013}.

\bibitem{Kingma}
{\sc D.~Kingma and J.~Ba}, {\em Adam: A method for stochastic optimization},
  International Conference on Learning Representations,  (2014).

\bibitem{LiunYong}
{\sc H.~Liu and J.~Yong}, {\em Option pricing with an illiquid underlying asset
  market}, Journal of Economic Dynamics \& Control, 29 (2005), pp.~2125--2156,
  \url{https://doi.org/http://apps.olin.wustl.edu/faculty/liuh/Papers/Liu_Yong.pd}.

\bibitem{Margrabe}
{\sc W.~Margrabe}, {\em The value of an option to exchange one asset for
  another}, Journal of Finance, 33 (1978), pp.~177--186,
  \url{https://doi.org/https://doi.org/10.2307/2326358}.

\bibitem{Milstein}
{\sc G.~N. Mil’shtein}, {\em Approximate integration of stochastic
  differential equations}, Theory of Probability \& Its Applications., 19
  (1975), pp.~557--000, \url{https://doi.org/https://doi.org/10.1137/1119062}.

\bibitem{Oksendal}
{\sc B.~Oksendal}, {\em Stochastic Differential Equations (3rd Ed.): An
  Introduction with Applications}, Springer-Verlag, Berlin, Heidelberg, 1992.

\bibitem{Pirvu2}
{\sc T.~Pirvu and A.~Yazdanian}, {\em Numerical analysis for spread option
  pricing model in illiquid underlying asset market: Full feedback model},
  Applied Mathematics \& Information Sciences, 10 (2015), pp.~1271--1281,
  \url{https://doi.org/10.18576/amis/100406}.

\bibitem{Rubin}
{\sc R.~Y. Rubinstein and R.~Marcus}, {\em Efficiency of multivariate control
  variates in monte carlo simulation}, Operations Research, 33 (1985),
  pp.~661--677, \url{https://doi.org/10.1287/opre.33.3.661},
  \url{https://doi.org/10.1287/opre.33.3.661}.

\bibitem{Scheicher}
{\sc K.~Scheicher}, {\em Complexity and effective dimension of discrete lévy
  areas}, Journal of Complexity, 23 (2007), pp.~152--168,
  \url{https://doi.org/doi:10.1016/j.jco.2006.12.006}.

\bibitem{Pirvu}
{\sc A.~Shidfar, K.~Paryab, A.~Yazdanian, and T.~A. Pirvu}, {\em Numerical
  analysis for spread option pricing model of markets with finite liquidity:
  first-order feedback model}, International Journal of Computer Mathematics,
  91 (2014), pp.~2603--2620,
  \url{https://doi.org/https://doi.org/10.1080/00207160.2014.887274}.

\bibitem{ShreveII}
{\sc S.~E. Shreve}, {\em Stochastic calculus for finance II, Continuous-time
  models}, Springer, New York, NY; Heidelberg, 2004.

\bibitem{Siraj}
{\sc S.~ul~Islam and I.~Ahmad}, {\em A comparative analysis of local meshless
  formulation for multi-asset option models}, Engineering Analysis with
  Boundary Elements, 65 (2016), pp.~159 -- 176,
  \url{http://www.sciencedirect.com/science/article/pii/S0955799716000175}.

\bibitem{Wilmott}
{\sc P.~Wilmott and P.~J. Schönbucher}, {\em The feedback effect of hedging in
  illiquid markets}, SIAM Journal on Applied Mathematics, 61 (2000),
  pp.~232--272, \url{https://doi.org/10.1137/S0036139996308534}.

\end{thebibliography}
\end{document}